\documentclass[a4paper,UKenglish,cleveref, autoref, thm-restate]{lipics-v2019}
\usepackage{microtype}
\usepackage[utf8]{inputenc}
\usepackage{amsmath, amsthm, amssymb, mathtools, stmaryrd}
\usepackage{tikz}
\usepackage{xspace}
\usepackage{todonotes}
\usepackage{enumerate}
\usepackage{listings}
\usepackage{algorithm}
\usepackage{comment}
\usepackage{stmaryrd}

\usetikzlibrary{backgrounds}
\usetikzlibrary{patterns}
\usetikzlibrary{decorations.pathreplacing,calligraphy}

\usepackage{algorithmic}

\makeatletter
\newcommand{\@abbrev}[3]{
	\def\c@a@def##1{
		\if ##1.
		\relax
		\else
		\@ifdefinable{\@nameuse{#1##1}}{\@namedef{#1##1}{#2##1}}
		\expandafter\c@a@def
		\fi
	}
	\c@a@def #3.
}
\@abbrev{bb}{\mathbb}{ABCDEFGHIJKLMNOPQRSTUVWXYZ}
\@abbrev{bf}{\mathbf}{ABCDEFGHIJKLMNOPQRSTUVWXYZabcdefghijklmnopqrstuvwxyz}
\@abbrev{bit}{\boldsymbol}{
	ABCDEFGHIJKLMNOPQRSTUVWXYZabcdefghijklmnopqrstuvwxyz}
\@abbrev{mc}{\mathcal}{ABCDEFGHIJKLMNOPQRSTUVWXYZ}
\@abbrev{mb}{\mathbb}{ABCDEFGHIJKLMNOPQRSTUVWXYZ}
\@abbrev{mf}{\mathfrak}{
	ABCDEFGHIJKLMNOPQRSTUVWXYZabcdefghijklmnopqrstuvwxyz}
\@abbrev{rm}{\mathrm}{ABCDEFGHIJKLMNOPQRSTUVWXYZabcdefghijklmnopqrstuvwxyz}
\@abbrev{scr}{\mathscr}{
	ABCDEFGHIJKLMNOPQRSTUVWXYZabcdefghijklmnopqrstuvwxyz}
\@abbrev{sf}{\mathsf}{ABCDEFGHIJKLMNOPQRSTUVWXYZabcdefghijklmnopqrstuvwxyz}
\@abbrev{b}{\bar}{abcdefghijklmnopqrstuvwxyz}
\makeatother

\newcommand{\FO}{\ensuremath{\textsc{FO}}}

\newcommand{\CPT}{\ensuremath{\textsc{CPT}}}

\newcommand{\classmonPCx}[1]{\ensuremath{\mathbf{\mathcal{K}_{\monPCx 
				k}}}}

\newcommand{\HF}{\text{HF}}
\newcommand{\Aut}{\mathbf{Aut}}

\newcommand{\Sym}{{\mathbf{Sym}}}

\renewcommand{\phi}{\varphi}

\newcommand{\lra}{\longrightarrow}

\renewcommand{\phi}{\varphi}
\renewcommand{\theta}{\vartheta}

\renewcommand{\AA}{{\mathfrak A}}

\renewcommand{\epsilon}{\varepsilon}

\newcommand{\tc}{\text{tc}}

\newcommand{\Ee}{{\cal E}}

\newcommand{\Kk}{{\cal K}}

\newcommand{\Oo}{{\cal O}}
\newcommand{\Pp}{{\cal P}}

\renewcommand{\bar}{\overline}

\newcommand{\cutout}[1]{}

\newcommand{\Stab}{\mathbf{Stab}}
\newcommand{\StabP}{\Stab^\bullet}
\newcommand{\Sp}{\mathbf{SP}}
\newcommand{\SpC}{\Sp}
\newcommand{\Orbit}{\mathbf{Orbit}}

\title{Choiceless Computation and Symmetry: Limitations of Definability} %TODO Please add

\titlerunning{Limitations of choiceless definability} %TODO optional, please use if title is longer than one line

\author{Benedikt Pago}{Mathematical Foundations of Computer Science, RWTH Aachen University, Germany }{pago@logic.rwth-aachen.de}{}{}%TODO mandatory, please use full name; only 1 author per \author macro; first two parameters are mandatory, other parameters can be empty. Please provide at least the name of the affiliation and the country. The full address is optional

\authorrunning{B. Pago} %TODO mandatory. First: Use abbreviated first/middle names. Second (only in severe cases): Use first author plus 'et al.'

\Copyright{Benedikt Pago} %TODO mandatory, please use full first names. LIPIcs license is "CC-BY";  http://creativecommons.org/licenses/by/3.0/


 %TODO mandatory; please add comma-separated list of keywords

\relatedversion{} %optional, e.g. full version hosted on arXiv, HAL, or other respository/website
\supplement{}%optional, e.g. related research data, source code, ... hosted on a repository like zenodo, figshare, GitHub, ...

\acknowledgements{I would like to thank my advisor Erich Grädel for helpful comments and discussions.}%optional

\nolinenumbers %uncomment to disable line numbering

\EventEditors{Christel Baier and Jean Goubault-Larrecq}
\EventNoEds{2}
\EventLongTitle{29th EACSL Annual Conference on Computer Science Logic (CSL 2021)}
\EventShortTitle{CSL 2021}
\EventAcronym{CSL}
\EventYear{2021}
\EventDate{January 25--28, 2021}
\EventLocation{Ljubljana, Slovenia (Virtual Conference)}
\EventLogo{}
\SeriesVolume{183}
\ArticleNo{33}

\begin{document}

\maketitle

\begin{abstract}
	The search for a logic capturing PTIME is a long standing open problem in finite model theory. One of the most promising candidate logics for this is \emph{Choiceless Polynomial Time} with counting (CPT). Abstractly speaking, CPT is an isomorphism-invariant computation model working with hereditarily finite sets as data structures.\\
	While it is easy to check that the evaluation of CPT-sentences is possible in polynomial time, the converse has been open for more than 20 years: Can every PTIME-decidable property of finite structures be expressed in CPT?\\
	We attempt to make progress towards a negative answer and show that Choiceless Polynomial Time cannot compute a \emph{preorder} with colour classes of \emph{logarithmic size} in every hypercube. The reason is that such preorders have super-polynomially many automorphic images, which makes it impossible for CPT to define them.\\
	While the computation of such a preorder is not a decision problem that would immediately separate P and CPT, it is significant for the following reason: The so-called Cai-Fürer-Immerman (CFI) problem is one of the standard ``benchmarks'' for logics and maybe best known for separating fixed-point logic with counting (FPC) from P. Hence, it is natural to consider this also a potential candidate for the separation of CPT and P. The strongest known positive result in this regard says that CPT is able to solve CFI if a preorder with logarithmically sized colour classes is present in the input structure.\\
	Our result implies that this approach cannot be generalised to unordered inputs. In other words, CFI on unordered hypercubes is a PTIME-problem which provably cannot be tackled with the state-of-the-art choiceless algorithmic techniques.
\end{abstract}

\section{Introduction}
\label{sec:introduction}
One of the big open questions in descriptive complexity theory is whether there exists a logic capturing PTIME (see \cite{chandra1980structure}, \cite{grohe2008quest}, \cite{gurevich1985logic}, \cite{pakusa2015linear}). Towards an answer to this question, several logics of increasing expressive power within PTIME have been devised, the best-studied of which is probably FPC, \emph{fixed-point logic with counting} (see \cite{dawar2015nature} for a survey). However, FPC only corresponds to a strict subset of PTIME because it cannot express the so-called \emph{CFI query}, a version of the graph isomorphism problem on certain graphs constructed by Cai, Fürer and Immerman in 1992 \cite{caifurimm92}. This problem is in P and has turned out to be extremely valuable as a benchmark for PTIME-logics as well as for certain classes of graph isomorphism algorithms.\\

The most important candidate logics for capturing PTIME, which have not yet fallen prey to the CFI problem, are \emph{Rank logic} \cite{dawar2009logics} and \emph{Choiceless Polynomial Time} (CPT). CPT was introduced in 1999 by Blass, Gurevich and Shelah \cite{blass1999} as a machine model that comes as close to Turing machines as possible, while enforcing \emph{isomorphism-invariance} of the computations -- this property is precisely the main difference between logics and classical Turing machines. Since its original invention, various different formalisations of CPT have emerged but the underlying principle is always the same: \emph{Symmetric} computation on \emph{polynomially-sized hereditarily finite sets} as data structures. \\

Not many lower bound results for Choiceless Polynomial Time are known so far, and of course, no \emph{decision problem} in P has been shown to be undefinable in CPT. However, what has been achieved is a non-definability statement for a \emph{functional problem}: Rossman showed that CPT cannot define the dual of any given finite vector space \cite{rossman2010choiceless}. Our contribution is a result of a similar kind, but stronger in a sense: We show non-definability not only for a concrete functionally determined object, but for all objects satisfying a certain set of properties. Concretely, no CPT program can define a hereditarily finite set representing a preorder with colour classes of logarithmic size in every hypercube. This can be seen -- potentially -- as a first step towards a non-definability result for a decision problem: the already mentioned CFI query; this would separate CPT from PTIME. Let us explain what undefinable preorders have to do with the CFI problem (see \autoref{sec:CFI} for details).\\

A preorder in a structure can be seen as a linear order on a collection of \emph{colour classes}, which form a partition of the universe: These colour classes are subsets of the structure whose elements are pairwise indistinguishable. The smaller the colour classes are, the ``finer'' is the preorder, and the more closely it resembles a linear order. By the famous Immerman-Vardi Theorem (\cite{immerman1982relational}, \cite{vardi1982complexity}), fixed-point logic, and therefore also CPT, captures PTIME on linearly ordered structures. Therefore, intuitively speaking, hard problems like CFI should become easier to handle if CPT is able to define a sufficiently fine preorder, or even a linear order, on the input structure. Indeed, Pakusa, Schalthöfer and Selman showed that CPT can define the CFI query if a preorder with colour classes of logarithmic size is available \cite{pakusa2018definability}. This is the strongest known positive result concerning the solvability of CFI in CPT.\\
Our contribution implies that this result cannot be generalised to the CFI problem on unordered input structures: Instances of CFI can be obtained by applying the Cai-Fürer-Immerman construction to any family of connected graphs, in particular also to hypercubes. Since CPT cannot define a sufficiently fine preorder in all hypercubes, and the CFI construction preserves the hypercube-structure, the algorithmic technique from \cite{pakusa2018definability} which heavily relies on  such preorders cannot be applied to all unordered CFI structures.\\
Therefore, if CFI on unordered structures is solvable in CPT, entirely new choiceless algorithmic techniques are needed to show this. Otherwise, if CFI is indeed a separating problem for CPT and P, one possible approach to prove this would be to identify further hereditarily finite sets over hypercubes which are not CPT-definable.\\

Technically, what we show in this paper is a statement concerning the orbit size of certain hereditarily finite objects over hypercubes: For every $n \in \bbN$, fix a h.f.\ object representing a preorder in the $n$-dimensional hypercube. If the colour classes of each preorder are of logarithmic size w.r.t.\ the hypercube, then the orbit size (w.r.t.\ the hypercube-automorphisms) of these h.f.\ objects grows super-polynomially in $2^n$, which is the size of the $n$-dimensional hypercube.\\

Since CPT is a logic and therefore isomorphism-invariant, it has to define any object together with its entire orbit -- if the size of the orbit is not polynomially bounded, then this is not possible in Choiceless Polynomial Time. In fact, we can interpret this non-definability result as an inherent weakness of choiceless polynomial time computation in general: It holds for any isomorphism-invariant polynomial time (or even polynomial space) computation model on hereditarily finite sets. Hence, should it be the case that CPT fails to capture PTIME because of a super-polynomial orbit argument like this one, we could conclude that the quest for a PTIME-logic should continue with other data structures than hereditarily finite sets.\\

Finally, we remark that the main combinatorial tool we use in our proof -- so-called \emph{supporting partitions} -- is taken from \cite{anderson2017symmetric}, where Anderson and Dawar show a correspondence between FPC and symmetric circuits. There, it is used for the calculation of orbit sizes of circuit gates. The fact that this tool also helps to understand the symmetries of hereditarily finite objects over hypercubes demonstrates its versatility and usefulness for the study of symmetric objects in general.

\section{Choiceless computation and the undefinability of preorders}
\label{sec:CPT}
In this paper, we will not give a definition of CPT, but only state its properties that our lower bound depends on. Thereby, our result also holds for a much broader class of choiceless computation models that includes CPT.\\
For details on CPT, we refer to the literature: A concise survey on the subject can be found in \cite{gradel2015polynomial}. It should be noted that there are multiple different ways to formalise CPT: The original definition was via abstract state machines \cite{blass1999}, but there are also more ``logic-like'' presentations such as Polynomial Interpretation Logic (see \cite{grapakschalkai15}, \cite{svenja}) and BGS-logic \cite{rossman2010choiceless}. The latter is essentially a fixed-point logic that allows for the isomorphism-invariant creation and manipulation of \emph{hereditarily finite sets} over the input structure. In fact, it has been shown in \cite{dawar2008choiceless} that any CPT-program (the words ``program'' and ``sentence'' are often used interchangeably in the context of CPT) is equivalent to a sentence in FPC (fixed-point logic with counting) evaluated in the input structure enriched with all the necessary hereditarily finite sets. Therefore, let us make this notion precise.

\paragraph*{Hereditarily finite sets and choiceless computation} Let $A$ be a non-empty set. The set of \emph{hereditarily finite objects} over $A$, $\HF(A)$, is defined as $\bigcup_{i \in \bbN} \HF_i(A)$, where $\HF_0(A) := A \cup \{\emptyset\}, \HF_{i+1}(A) := \HF_{i}(A) \cup 2^{\HF_{i}(A)}$. The size of an h.f.\ set $x \in \HF(a)$ is measured in terms of its \emph{transitive closure} $\tc(x)$: The set $\tc(x)$ is the least transitive set such that $x \in \tc(x)$. Transitivity means that for every $a \in \tc(x)$, $a \subseteq \tc(x)$.\\ 
If the atom set $A$ is the universe of a structure $\AA$, then the action of $\Aut(\AA) \leq \Sym(A)$, the automorphism group of $\AA$, extends naturally to $\HF(A)$: For $x \in \HF(A)$, $\pi \in \Aut(\AA)$, $\pi(x)$ is obtained from $x$ by replacing each occurrence of an atom $a$ in $x$ with $\pi(a)$.\\
The \emph{orbit} (w.r.t.\ the action of $\Aut(\AA)$) of an object $x \in \HF(A)$ is the set of all its automorphic images, i.e.\ $\{\pi(x) \mid \pi \in \Aut(\AA)\}$. The \emph{stabiliser} $\Stab(x)$ of $x$ is the subgroup $\{\pi \in \Aut(\AA) \mid \pi(x) = x \}$.

\begin{definition}
	\label{def:symmetry}
	Let $\AA$ be a finite relational structure with universe $A$, and $p : \bbN \lra \bbN$ a polynomial. We say that a h.f.\ object $x \in \HF(A)$ is
	\begin{itemize}
		\item \emph{symmetric} (w.r.t.\ $\AA$) if $x$ is stabilised by all automorphisms of $\AA$, i.e.\ $\Stab(x) = \Aut(\AA)$;
		\item \emph{$p$-bounded} if $|\tc(x)| \leq p(|A|)$.
	\end{itemize}
\end{definition} 

Every CPT-program comes with an explicit polynomial bound $p$ that limits both the length of its runs as well as the size of the h.f.\ sets that it may use in the computation. Further, due to its nature as a logic, all operations of CPT are symmetry-invariant. This is already everything that our lower bound depends on. The following abstract view on the execution of CPT-programs is true regardless of the concrete presentation of CPT, and this level of abstraction is sufficient for the purposes of this paper:\\  

Let $\Pi$ be a CPT-program with bound $p$, and $\AA$ be a structure of matching signature. Then the run of $\Pi$ on $\AA$ is a sequence of h.f.\ sets $x_1, x_2, ... \in \HF(A)$, each of which is symmetric and $p$-bounded w.r.t.\ $\AA$.\\

Consequently, no CPT-program -- and generally, no computation model operating on symmetric $p$-bounded h.f.\ sets -- can compute a h.f.\ set $x$ with super-polynomial orbit size because the corresponding stage of the run must contain $x$ along with its entire orbit in order to fulfil the symmetry-condition. Now, we are almost ready to state our general lower bound theorem, which applies to CPT as a special case by the facts just mentioned.

\paragraph*{Preorders and colour classes}
A \emph{preorder} $\prec$ on a set $A$ induces a partition of $A$ into \emph{colour classes} $C_1,...,C_m$. A colour class is a set of $\prec$-incomparable elements, and $\prec$ induces a linear order on the colour classes. Our technical contribution is to show that any such ordered partition with colour classes of logarithmic size in the family of $n$-dimensional hypercubes has super-polynomial orbit size. From this it follows that such preorders are not CPT-definable in the hypercubes, regardless of how they are represented as h.f.\ sets. For a structure $\AA$ with universe $A$, we say that a set $x \in \HF(A)$ \emph{encodes} an ordered partition over $\AA$ if there exists a CPT-program which computes the following natural representation of $(C_1,...,C_m)$ on input $\AA$ and $x$: $\{ C_1, \{ C_2, \{ C_3, \{  ...  \} \} \}    \}$.
If the ordered partition $\Pp_n = (C_1,...,C_m)$ has a super-polynomially large orbit, then the same is true for any h.f.\ set that encodes $\Pp_n$ (see Proposition \ref{prop:encodingsHaveLargeOrbit}).
Thus, our main result is summarised as follows:
\begin{theorem}
	\label{thm:mainResult}
	Let $(H_n)_{n \in \bbN}$ be the family of $n$-dimensional hypercubes. For each $n \in \bbN$, fix an ordered partition $\Pp_n = (C_1,...,C_{m_n})$ of $\{0,1\}^n$ such that each part $C_i$ has size at most $\Oo(n) = \Oo(\log |H_n|)$. Let $x_n$ be any \emph{symmetric} (w.r.t.\ $H_n$) h.f.\ set over $H_n$ that \emph{encodes} $\Pp_n$. Then there exists no polynomial $p$ such that $x_n$ is also \emph{$p$-bounded} w.r.t.\ $H_n$. 
\end{theorem}
The proof can be found in \autoref{sec:main}.
As already explained, this implies the following nondefinability statement for CPT.
\begin{corollary}
	There is no CPT-program that computes in every hypercube $H_n$ a (h.f.\ set representation of a) total preorder with colour classes of size $\Oo(n)= \Oo(\log |H_n|)$.
\end{corollary}

\section{Previous work and the significance of undefinable preorders}
\label{sec:CFI}
As already mentioned, our contribution is a non-definability result for a \emph{functional problem}, the computation of certain preorders.\\
However, our research is motivated by the study of a \emph{decision problem} which is seen as a potential candidate for the separation of CPT from PTIME: The so-called \emph{CFI problem}, that we briefly introduce next. Whether CFI in its general version is solvable in $\CPT$ is an open question, but at least for restricted versions, where the structures possess a certain degree of built-in order, it is known to be in CPT. Our non-definability result implies that being able to solve the restricted version of CFI in CPT is of no help for solving CFI in the general case.
\paragraph*{The CFI problem}
For a detailed account of the CFI problem and the construction of the so-called CFI graphs, we refer the reader to the original paper \cite{caifurimm92} by Cai, Fürer and Immerman. Here, we only review it to an extent sufficient for our purposes.\\

Essentially, CFI is the Graph Isomorphism problem on specific pairs of graphs that are obtained by applying the so-called CFI construction to a family of connected graphs, for example, to hypercubes. These are referred to as the \emph{underlying graphs}. The construction replaces every edge and every vertex of the underlying graph with a gadget. Importantly, the symmetries of the underlying graph are preserved this way.
Any underlying graph $G$ can be transformed into an odd and an even CFI graph, $G_0$ and $G_1$. It holds $G_0 \not\cong G_1$, and there is a simple polynomial time algorithm which can determine, given a CFI graph $G_x$, whether it is odd or even, i.e.\ if $G_x \cong G_0$, or $G_x \cong G_1$. This is what the CFI problem asks for.\\
However, on the logical side, that is, in $\FO$ with counting, $G_0$ and $G_1$ can only be distinguished with a linear number of variables. As a consequence, no FPC-sentence can solve the CFI-problem (on a suitable class of underlying graphs). Since this very expressive ``reference logic'' within PTIME fails to solve CFI, this raises the question whether CPT is strong enough to achieve this, or if CFI is indeed a problem that separates CPT from P.
\paragraph*{Solving CFI in CPT}
If the underlying graphs of the CFI construction satisfy certain properties, then CFI can be solved in CPT:
\begin{theorem}[\cite{pakusa2018definability}]
	\label{thm:svenjaPreorder}
	Let $\Kk$ be the class of connected, preordered graphs $G = (V,E,\prec)$ where the size of each colour class is bounded by $\log |V|$. The CFI problem on underlying graphs in $\Kk$ can be solved in Choiceless Polynomial Time.
\end{theorem}	
This is the strongest known positive result concerning CFI and CPT. It is a generalisation of the CPT-algorithm by Dawar, Richerby and Rossman from \cite{dawar2008choiceless} for the CFI problem on linearly ordered graphs. Note that not the CFI graphs $G_0, G_1$ are ordered/preordered in these settings, but only the underlying graph $G$ (otherwise, the Immerman-Vardi Theorem could be applied). The order/preorder on $G$ allows for the algorithmic creation of a so-called ``super-symmetric'' h.f.\ object with polynomial orbit which makes it possible to determine the parity of the input CFI graph $G_x$. This object reflects in its structure the preorder on the input, and is therefore not definable in unordered inputs according to our \autoref{thm:mainResult}: It can be checked that \autoref{thm:mainResult} not only holds for hypercubes but also for the CFI graphs obtained from them; this is true because $\Aut(G)$ embeds into $\Aut(G_i)$ for any graph $G$ and corresponding CFI graph $G_i$. Hence, the algorithmic technique that proves \autoref{thm:svenjaPreorder} cannot be generalised to the CFI problem on unordered graphs. In fact, any CPT algorithm that is to solve the unordered CFI problem must avoid the construction of a h.f.\ object whose nesting structure induces a too fine preorder on the input.\\

We remark that there are of course families of graphs where the undefinability of such preorders is much easier to show than on hypercubes. For instance, on complete graphs, it is clear that the orbit of a preorder with logarithmic colour classes grows super-polynomially. However, the size of any CFI graph $G_0$ is exponential in the maximal degree of $G$. Therefore, the polynomial resources of CPT suffice to solve CFI on unordered graphs of linear maximal degree (this is another result from \cite{pakusa2018definability}). Hence, complete graphs do not yield hard CFI instances. In contrast, CFI on hypercubes is well-suited as a benchmark for CPT because their degree is logarithmic and thus the CFI construction only increases the size polynomially.\\

Our lower bound is a first piece of evidence that the CFI problem on hypercubes is hard (and maybe even unsolvable) for CPT and we believe that it deserves further investigation. The results in \cite{dawar2008choiceless} indirectly suggest a systematic way to do so: Namely, Dawar, Richerby and Rossman showed that -- as long as the CFI structures satisfy a certain homogeneity condition -- solving the CFI problem in CPT always requires the construction of a h.f.\ set which contains a large subset of the input structure as atoms. If it were possible to show that no sufficiently large h.f.\ object over hypercubes has a polynomial orbit, then this could be used to separate CPT from PTIME. Our result is a step in that direction as it suggests that this large object cannot be structurally similar to a preorder.

\section{Analysing orbits of hereditarily finite objects over hypercubes}
\label{sec:hypercubes}

Let $H_n = (V_n, E_n)$ be the $n$-dimensional hypercube, i.e.\ $V_n = \{0,1\}^n$,\\ $E_n = \{\{u,v\} \in V_n^2 \mid d(u,v) = 1 \}$, where $d(u,v)$ is the Hamming-distance.\\ 
Its automorphism group $\Aut(H_n)$ is isomorphic to the semi-direct product of $\Sym_n$ and $(\{0,1\}^n,\oplus)$ \cite{harary2000automorphism}, where $\Sym_n$ is the symmetric group on $[n] = \{1,2,...,n\}$, and $(\{0,1\}^n,\oplus)$ is the group formed by the length-$n$ binary strings together with the bitwise XOR-operation. More precisely, any automorphism $\sigma \in \Aut(H_n)$ corresponds to the pair $(\pi,w) \in \Sym_n \times \{0,1\}^n$ with $\sigma(v) = \pi(v) \oplus w$, where $\pi(v) = v_{\pi^{-1}(1)}v_{\pi^{-1}(2)}...v_{\pi^{-1}(n)}$ (i.e.\ $\pi(v)$ is obtained from $v$ by permuting the positions of the word according to $\pi$). This means: $|\Aut(H_n)| = n! \cdot 2^n$. Note that it is the factor $n!$ which makes the size of this group super-polynomial in $|V_n| = 2^n$.\\

Our main technical theorem, \autoref{thm:mainTechnical} concerns a fixed sequence of ordered partitions of the $n$-dimensional hypercubes, $(\Pp_n)_{n \in \bbN}$. We aim for a lower bound on the orbit size of these ordered partitions $\Pp_n$ w.r.t.\ the action of $\Aut(H_n)$. For our purposes, it only matters whether this lower bound is super-polynomial in $2^n = |V_n|$, or not. For this question, we can restrict ourselves to automorphisms corresponding to permutation-word pairs of the form $(\pi,0^n)$, for $\pi \in \Sym_n$. Therefore, for the rest of this paper, we simply let $\Sym_n$ act on $V_n$ by permuting the positions of the binary strings as described above. In this sense, $\Sym_n$ embeds into $\Aut(H_n)$, and hence, whenever $\Pp_n$ has a super-polynomial orbit with respect to this action of $\Sym_n$, this is also true with respect to the action of $\Aut(H_n)$.\\

To sum up, our task is to lower-bound the orbit-sizes of ordered partitions $\Pp$ of $\{0,1\}^n$ with respect to $\Sym_n$ acting on the positions of the strings. We do this via the Orbit-Stabiliser Theorem. 
Let $\Stab_n(\Pp)$ and $\Orbit_n(\Pp)$ denote the stabiliser and orbit, respectively, of $\Pp$ w.r.t.\ the action of $\Sym_n$ on the string positions:
\[
\Stab_n(\Pp) := \{ \pi \in \Sym_n \mid \pi(\Pp) = \Pp \}, \
\Orbit_n(\Pp) := \{\pi(\Pp) \mid \pi \in \Sym_n\}.
\]

\begin{proposition}[Orbit-Stabiliser]
\[
|\Orbit_n(\Pp)|= \frac{|\Sym_n|}{|\Stab_n(\Pp)|} = \frac{n!}{|\Stab_n(\Pp)|}.
\]
\end{proposition}
This means that we have to upper-bound $|\Stab_n(\Pp)|$. Since the partition is viewed as an ordered tuple of parts, any automorphism that stabilises $\Pp = (C_1,...,C_m)$ must stabilise each $C_i$ (not necessarily pointwise, but as a set):

\begin{proposition}
	\label{prop:basicFact}
	\[
	\Stab_n(\Pp) \subseteq \bigcap_{C \in \Pp} \Stab_n(C).
	\]
\end{proposition}
In other words, we have reduced our problem to upper-bounding the size of the simultaneous stabiliser group of a collection of sets of bitstrings. In the next section, we introduce a tool that we need in order to accomplish this: So-called \emph{supporting partitions}.

\section{Approximating permutation groups with supporting partitions}
\label{sec:support}
The notions and results in this section are mostly taken from the paper on Symmetric Circuits and FPC by Anderson and Dawar \cite{anderson2017symmetric}. 

\begin{definition}
	Let $\Pp$ be a partition of $[n]$. 
	\begin{itemize}
		\item The \emph{pointwise} stabiliser of $\Pp$ is $\StabP_n(\Pp) := \{\pi \in \Sym_n \mid \pi(P) = P \text{ for all } P \in \Pp \}.$
		\item The \emph{setwise} stabiliser of $\Pp$ is $\Stab_n(\Pp) := \{\pi \in \Sym_n \mid \pi(P) \in \Pp \text{ for all } P \in \Pp \}$ (these are all $\pi \in \Sym_n$ that induce a permutation on the parts of $\Pp$).
	\end{itemize}	
\end{definition}

\begin{definition}[Supporting Partition, \cite{anderson2017symmetric}]
	Let $G \leq \Sym_n$ be a group. A supporting partition $\Pp$ of $G$ is a partition of $[n]$ such that $\StabP_n(\Pp) \leq G$.
\end{definition}
A group $G \subseteq \Sym_n$ may have several supporting partitions but there always exists a unique \emph{coarsest supporting partition}. A partition $\Pp'$ is as coarse as a partition $\Pp$, if every part in $\Pp$ is contained in some part in $\Pp'$. For any two partitions $\Pp, \Pp'$ there exists a finest partition $\Ee(\Pp,\Pp')$ that is as coarse as either of them:

\begin{definition}[\cite{anderson2017symmetric}]
	Let $\Pp,\Pp'$ be partitions of $[n]$. Let $\sim$ be a binary relation on $[n]$ such that $x \sim y$ iff there exists a part $P \in \Pp$ or $P \in \Pp'$ such that $x,y \in P$. Then $\Ee(\Pp,\Pp')$ is the partition of $[n]$ whose parts are the equivalence classes of $[n]$ under the transitive closure of $\sim$.
\end{definition}
As shown in \cite{anderson2017symmetric}, the property of being a supporting partition of a group $G \subseteq \Sym_n$ is preserved under the operation $\Ee$. Therefore it holds:

\begin{lemma}[\cite{anderson2017symmetric}]
	Each permutation group $G \leq \Sym_n$ has a unique \emph{coarsest supporting partition}, denoted $\Sp(G)$.
\end{lemma}
When we write $\Sp(a)$ for $a \in \HF(\{0,1\}^n)$, we mean $\Sp(\Stab_n(a))$, that is, the coarsest supporting partition of the stabiliser of $a$, where -- as in the previous section -- we consider the stabiliser as the subgroup of $\Sym_n$ acting on the positions of the binary strings. Note that if $a \in \{0,1\}^n$, then $\Sp(a)$ is just the partition of $[n]$ into $\{ k \in [n] \mid a_k = 0 \}$ and $\{ k \in [n] \mid a_k = 1 \}$.\\
\\
The reason why coarsest supporting partitions are useful for estimating the sizes of certain stabiliser subgroups is the following result:
\begin{lemma}[\cite{anderson2017symmetric}]
	\label{lem:sandwichLemma}
	Let $G \leq \Sym_n$ be a group. Then:
	\[\StabP_n(\Sp(G)) \leq G \leq \Stab_n(\Sp(G)).\]
\end{lemma}
This lemma enables us to upper-bound stabilisers of subsets of $\{0,1\}^n$ by the stabilisers of their supporting partitions.\\
Finally, because we will frequently need it later in our proof, we define the operation $\sqcap$ as the ``intersection'' of two partitions of $[n]$:
\begin{definition}[Intersection of partitions]
	\label{def:intersection}
	Let $\Pp,\Pp'$ be partitions of $[n]$. The \emph{intersection} $\Pp \sqcap \Pp'$ is defined like this:
	\[\Pp \sqcap \Pp' := \{ \Pp(k) \cap \Pp'(k) \mid k \in [n] \}.\]
	Here, $\Pp(k), \Pp'(k)$ denote the parts of the respective partition that contain $k$.
\end{definition}

\section{The Super-Polynomial Orbit Theorem}
\label{sec:main}
Our main technical theorem reads as follows: 
\begin{theorem}
	\label{thm:mainTechnical}
Let $(\Pp_n)_{(n \in \bbN)}$ be a sequence where each $\Pp_n = (C_1,...,C_{m_n})$ is an ordered partition of $\{0,1\}^n$. Assume that $\max_{i \in [m_n]}{|C_i|}$ is in $\Oo(n)$.
~\\
Then, $|\Orbit_n(\Pp_n)|$ (as defined in \autoref{sec:hypercubes}) grows asymptotically faster than any polynomial in $2^n = |\{0,1\}^n|$.	
\end{theorem} 
From this, \autoref{thm:mainResult} follows because 
any h.f.\ set that encodes a preorder with colour classes $\Pp_n$ then also has super-polynomial orbit size in $|H_n|$:
\begin{proposition}
	\label{prop:encodingsHaveLargeOrbit}
	If $|\Orbit_n(\Pp_n)|$ grows super-polynomially in $|H_n|$, then the orbit-size of every $x_n \in \HF(H_n)$ that encodes $\Pp_n$ over $H_n$ also grows super-polynomially in $|H_n|$.
\end{proposition}	
\textit{Proof sketch.} Suppose for a contradiction that $x_n$ encodes $\Pp_n$ and has an orbit of polynomial size. Then by our notion of encoding (see Section \ref{sec:CPT}), there exists a CPT-program that computes the representation $\{ C_1, \{ C_2, \{ C_3, \{  ...  \} \} \}   \}$ on input $H_n$ and $x_n$. But this set has exactly the symmetries of the ordered partition $\Pp_n$. Since CPT-computations do not break symmetries, it would then follow that also $\Pp_n$ has an orbit of polynomial size. \hfill \qedsymbol \\

Any \emph{symmetric} (see \autoref{def:symmetry}) h.f.\ object that contains some $x_n$ encoding $\Pp_n$ must necessarily contain $\Orbit_n(x_n)$, too. And then, this h.f.\ object is not $p$-bounded with respect to $|H_n| = 2^n$ because of Theorem \ref{thm:mainTechnical} and Proposition \ref{prop:encodingsHaveLargeOrbit}. This proves \autoref{thm:mainResult}.\\

We start to explain the proof idea of \autoref{thm:mainTechnical} by stating the following summary of \autoref{prop:basicFact} and \autoref{lem:sandwichLemma}:
\begin{corollary}
\label{cor:basicFact}
\[
\Stab_n(\Pp_n) \leq \bigcap_{C \in \Pp_n} \Stab_n(C) \leq \bigcap_{C \in \Pp_n} \Stab_n(\Sp(C)).
\]
\end{corollary}
We are going to employ the Orbit-Stabiliser Theorem in order to obtain our lower bound for the orbit size. Hence, we need to bound $|\Stab_n(\Pp_n)|$ from above, and \autoref{cor:basicFact} already indicates the basic principle of our proof: Splitting up $\Pp_n$ into its colour classes and analysing the stabilisers of their respective supporting partitions.\\
Our analysis of $|\Stab_n(\Pp_n)|$ is divided into two main cases that we treat separately. The distinction is with respect to the maximum size of the coarsest support of any part of $\Pp_n$, viewed as a function of $n$:\\
\\
Let $B_n \subseteq \{0,1\}^n$ denote the part $C$ of $\Pp_n$ such that $|\Sp(B_n)|$ (i.e.\ its number of parts) is maximal in $\{|\Sp(C)| \mid C \in \Pp_n \}$. Then the two cases we distinguish are:\\
\begin{enumerate}[(1)]
	\item The maximal level-support size grows sublinearly: $|\Sp(B_n) | \in o(n)$.
	\item The maximal level-support size grows linearly: $|\Sp(B_n) | \in \Theta(n)$.
\end{enumerate}

We deal with the two cases in the next two subsections. Their results are summarised in \autoref{lem:resultCaseSublinear} and \autoref{lem:LOG_caseLinearSupport}. Together they imply the theorem. In the following lemmas, we always refer to the setting of \autoref{thm:mainTechnical}, as well as to $B_n$ as just defined here.

\subsection{The case of sublinearly bounded supports}
 The result of this subsection is:
 
 \begin{lemma}
   \label{lem:resultCaseSublinear}
   	Assume the following three conditions hold:
   \begin{enumerate}
   	\item Every $v \in \{0,1\}^n$ occurs in at least one of the colour classes $C_i$ in $\Pp_n$.
   	\item The function $\max_{i \in [m_n]}{|C_i|}$ is in $\Oo(n)$.
   	\item $|\Sp(B_n)| \in o(n)$.
   \end{enumerate}
   Then the orbit size of $\Pp_n$ w.r.t.\ the action of $\Sym_n$ on $\{0,1\}^n$ grows faster than any polynomial in $2^n$.
 \end{lemma}
It may seem somewhat superfluous to state the first condition explicitly, since we are assuming that $\Pp_n$ is an ordered partition of $\{ 0,1\}^n$, so this is always satisfied. However, we would like to stress that our result crucially depends on this fact and thus, the proof will not directly go through if we replace the preorder by, say, a partial one that only partitions a small subset of $\{0,1\}^n$ into colour classes.\\
The second condition says that the colour classes of the preorder have logarithmic size (compared to the size of the hypercube) and is also required in Theorem \ref{thm:mainTechnical}. The third condition is specific to the case that we treat in this subsection.\\
\\	
By Corollary \ref{cor:basicFact}, we have for the stabiliser of $\Pp_n$:  
\[
\Stab(\Pp_n) \leq \bigcap_{C \in \Pp_n} \Stab(\Sp(C)).
\]
Let us briefly outline how we will use this fact to bound  $|\Stab(\Pp_n)|$. For each colour class $C \in \Pp_n$, $\Sym(\Sp(C))$ denotes the symmetric group on the parts of the supporting partition $\Sp(C)$ (in constrast, $\Sym_n$ is the symmetric group on the set $[n]$ that underlies this partition). Every $\pi \in \Sym_n$ that stabilises $\Sp(C)$ as a set \emph{induces} (or \emph{realises}) a $\sigma \in \Sym(\Sp(C))$ in the sense that $\sigma(P) = \{\pi(k) \mid k \in P \}$ for all $P \in \Sp_i(x_n)$. This can also be extended to a set $J_n \subseteq \Pp_n$ of several colour classes: Every $\pi \in \bigcap_{C \in J_n} \Stab(\Sp(C))$ induces a $\bar{\sigma} \in \bigtimes_{C \in J_n} \Sym(\Sp(C))$. Here, $\bar{\sigma}$ is the tuple of permutations that $\pi$ realises simultaneously on the parts of the respective supporting partitions.\\

Now in order to bound $|\Stab(\Pp_n)|$, we will choose a subset of colour classes $J_n \subseteq \Pp_n$ with certain properties that will enable us to bound two quantities: Firstly, there will be an upper bound on the number of $\bar{\sigma} \in \bigtimes_{C \in J_n} \Sym(\Sp(C))$ that can be realised by any $\pi \in \Stab(\Pp_n)$. Secondly, we observe that each such $\bar{\sigma}$ can only be realised by a small number of distinct $\pi \in \Stab(\Pp_n)$. The product of these two bounds is then an upper bound for $|\Stab(\Pp_n)|$.\\

We start the proof of Lemma \ref{lem:resultCaseSublinear} with the second part of the above proof sketch, namely we show how we can generally bound the number of realisations of a given $\bar{\sigma} \in \bigtimes_{i \in [m]} \Sym(\Sp(A_i))$, where $A_1,...,A_m$ are sets of length-$n$ bit-strings (in our case, the colour classes of $\Pp_n$). We refer back to Definition \ref{def:intersection} for the definition of $\sqcap$, the intersection of partitions. The following lemma shows that any $\pi \in \Sym_n$ that realises the mentioned tuple of permutations $\bar{\sigma}$ of the parts of the supports has to permute the parts of the intersection partition $\bigsqcap_{i = 1}^{m} \Sp(A_i)$ in the same way.
\begin{lemma}
	\label{lem:LOG_centralLemmaForMultipleSets}
	Let $A_1,...,A_m \subseteq \{0,1\}^n$ be a collection of sets of bit-strings. Fix any simultaneous permutation $\bar{\sigma}$ of the parts of the supports of the sets, i.e.\ $\bar{\sigma} \in \bigtimes_{i = 1}^{m} \Sym(\Sp(A_i))$.\\
	\\	
	There exists a $\theta_{\bar{\sigma}} \in \Sym( \bigsqcap_{i = 1}^{m} \Sp(A_i))$ such that every $\pi \in \Sym_n$ that realises $\bar{\sigma}$ also realises $\theta_{\bar{\sigma}}$.
\end{lemma}
\begin{proof}
	We show the statement via induction on $m$. For $m = 1$, there is only one set $A_1$, so $\bar{\sigma} \in \Sym(\Sp(A_1))$. Hence, $\theta_{\bar{\sigma}} := \bar{\sigma}$ is the desired permutation.\\
	For the inductive step, assume the statement holds for a fixed $m$. Consider now a collection $A_1,...,A_{m+1} \subseteq \{0,1\}^n$, and $\bar{\sigma} \in \bigtimes_{i = 1}^{m+1} \Sym(\Sp(A_i))$. Let $\bar{\sigma}'$ be the $m$-tuple that is the restriction of $\bar{\sigma}$ to its first $m$ components. By the induction hypothesis there exists a $\theta_{\bar{\sigma}'} \in \Sym( \bigsqcap_{i = 1}^{m} \Sp(A_i))$ that is induced by every $\pi \in \Sym_n$ realising $\bar{\sigma}'$, and therefore also by every $\pi$ realising $\bar{\sigma}$. Furthermore, every $\pi \in \Sym_n$ that realises $\bar{\sigma}$ has to realise $\bar{\sigma}_{m+1} \in \Sym(\Sp(A_{m+1}))$. So we know the following constraints for $\pi$: For every $P \in \Sp(A_{m+1})$, $\pi(P) = \bar{\sigma}_{m+1}(P)$, and also, for every $P \in \bigsqcap_{i = 1}^{m} \Sp(A_i)$, $\pi(P) = \theta_{\bar{\sigma}'}(P)$. Hence, the desired $\theta_{\bar{\sigma}} \in \Sym( \bigsqcap_{i = 1}^{m+1} \Sp(A_i))$ is defined as follows: Let $P \in \bigsqcap_{i = 1}^{m+1} \Sp(A_i)$, and call $Q_P \in \bigsqcap_{i = 1}^{m} \Sp(A_i)$ the part such that $P \subseteq Q_P$, and let $Q'_P \in \Sp(A_{m+1})$ be such that $P \subseteq Q'_P$ (the parts $Q_P, Q'_P$ must exist because the partition $\bigsqcap_{i = 1}^{m+1} \Sp(A_i)$ refines both $\bigsqcap_{i = 1}^{m} \Sp(A_i)$ and $\Sp(A_{m+1}$)). Then:
	\[
	\theta_{\bar{\sigma}}(P) := \theta_{\bar{\sigma}'}(Q_P) \cap \bar{\sigma}_{m+1}(Q'_P).
	\]
	One can easily check that indeed, $\theta_{\bar{\sigma}} \in \Sym(\bigsqcap_{i = 1}^{m+1} \Sp(A_i))$, and, as we argued already, that every $\pi \in \Sym_n$ that realises $\bar{\sigma}$ must realise $\theta_{\bar{\sigma}}$ as well. Below is a visualisation for the first inductive step, that adds $A_2$ to $A_1$.
\end{proof}

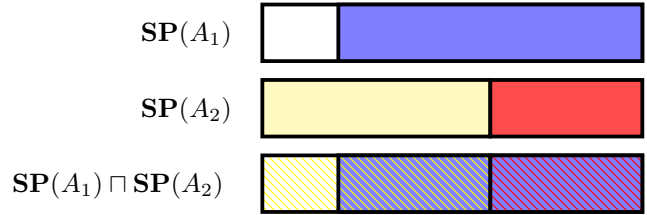
\begin{figure}[H]
	\centering
	\begin{tikzpicture}[line width = 0.5mm]
		\node at (-1,-0.4) {$\Sp(A_1)$};
		\node at (-1,-1.4) {$\Sp(A_2)$};
		\node at (-1.9,-2.4) {$\Sp(A_1) \sqcap \Sp(A_2)$};
		\draw (0,0) rectangle (5,-0.75);
		\draw[fill = blue!50] (1,0) rectangle (5,-0.75);
		
		\draw[fill = yellow!30] (0,-1) rectangle (3,-1.75);
		\draw[fill = red!70] (3,-1) rectangle (5,-1.75);
		
		\draw[pattern = north west lines, pattern color = yellow]  (0,-2) rectangle (5,-2.75);	
		\draw[preaction={fill = blue!50}, pattern = north west lines, pattern color = yellow] (1,-2) rectangle (3,-2.75);
		\draw[preaction={fill = blue!50}, pattern = north west lines, pattern color = red] (3,-2) rectangle (5,-2.75);
	\end{tikzpicture}
	\caption{Suppose $\bar{\sigma}$ is the identity permutation in both $\Sp(A_1)$ and $\Sp(A_2)$. Then any $\pi \in \Sym_n$ realising $\bar{\sigma}$ must stabilise each of the three parts of $\Sp(A_1) \sqcap \Sp(A_2)$ setwise. }
\end{figure}

So, intuitively speaking, the finer the partition $\bigsqcap_{i = 1}^m A_i$ is, the fewer realisations there are for any $\bar{\sigma} \in \bigtimes_{i = 1}^{m} \Sym(\Sp(A_i))$. Therefore, we will aim to select a subset of the colour classes in $\Pp_n$ such that the intersection over the supports is as fine as possible. More precisely, we would like it to consist of many singleton parts.\\
For the rest of this subsection we denote by $S_n \subseteq [n]$ the set of positions which are in singleton parts in  $\bigsqcap_{C \in \Pp_n} \Sp(C)$, i.e.\
\[
S_n := \{k \in [n] \mid \{k\} \in \bigsqcap_{C \in \Pp_n} \Sp(C)  \}.
\]
It turns out that there can only be few positions which are \emph{not} in singleton parts in $\bigsqcap_{C \in \Pp_n} \Sp(C)$; this is a consequence of the assumption that every element of $\{0,1\}^n$ occurs in a colour class of $\Pp_n$, together with the size bound on the colour classes (the first and second condition of Lemma \ref{lem:resultCaseSublinear}):
\begin{lemma}
	\label{lem:nonSingletonsUpperBound}
	Let $c$ be a constant such that for every colour class $C \in \Pp_n$, for large enough $n \in \bbN$, it holds: $|C| \leq c \cdot n$. Further, assume that every element of $\{0,1\}^n$ occurs in at least one $C \in \Pp_n$. Then, for all large enough $n$:
	\[
	|[n] \setminus S_n | <  8\log n.
	\]	 
\end{lemma}
\begin{proof}
	Assume for a contradiction:  $|[n] \setminus S_n | \geq 8\log n$. Let 
	\[
	\mathbf{P_2} := \{ P \in \bigsqcap_{C \in \Pp_n} \Sp(C) \mid |P| \geq 2 \}.
	\]
	Choose a string $a \in \{0,1\}^n$ such that its substring at the positions in $P$, denoted $a[P]$, contains an equal number of zeros and ones (or almost equal if $|P|$ is odd) for each $P \in \mathbf{P_2}$. Let $A_n$ denote the colour class $C \in \Pp_n$ such that $a \in C$. This exists by the assumptions of the lemma. For each $P \in\mathbf{P_2}$, a superset of $P$ (or $P$ itself) must occur as a part of $\Sp(A_n)$. We conclude that $\Stab(A_n)$ must contain all permutations $\pi \in \Sym_n$ which are the identity on $S_n$, and arbitrarily permute the elements within each part $P \in \mathbf{P_2}$ (due to Lemma \ref{lem:sandwichLemma}). So let 
	\[
	\Gamma_{\mathbf{P_2} } := \{ \pi \in \Sym_n \mid \pi(s) = s \text{ for all } s \in S_n \text{ and } \pi(P) = P \text{ for all } P \in \mathbf{P_2}\} \leq \Stab(A_n)
	\]
	be this subgroup of $\Sym_n$.\\
	All images of $a$ under the permutations in $\Gamma_{\mathbf{P_2} }$  must also be in $A_n$. This entails a violation of the size bound on $|A_n|$, as we show now. It is easy to see (observing that within every $P \in \mathbf{P_2}$, the $|P|/2$ many ones in $a[P]$ can be moved to an arbitrary subset of the positions $P$) that the orbit of $a$ with respect to $\Gamma_{\mathbf{P_2}}$ has size at least
	\begin{align*}
		\Orbit_{\Gamma_{\mathbf{P_2}}}(a) &\geq \prod_{P \in \mathbf{P_2}} \binom{|P|}{|P|/2} \geq \prod_{P \in \mathbf{P_2}} \delta \cdot \frac{2^{|P|}}{\sqrt{|P|}} \\
		&= 2^{|[n] \setminus S_n|} \cdot \prod_{P \in \mathbf{P_2}}  \frac{\delta}{\sqrt{|P|}} \tag{$\star$}
	\end{align*}
	Here, $0.6 \leq \delta < 1$. The inequality is quite well-known and can be computed with Stirling's approximation. The equality is clear because the parts in $\mathbf{P_2}$ cover exactly the positions $[n] \setminus S_n$.\\
	As for the large product, we can check that its value becomes smallest possible if all parts $P \in \mathbf{P_2}$ are doubleton parts (plus potentially one part with three elements). To see this, take any part $P \in \mathbf{P_2}$ with $|P| \geq 4$ and split off two elements, such that $P = P_1 \dot{\cup} P_2$, $|P_1| = 2$. Then the contribution of $P$ changes from $\frac{\delta}{\sqrt{|P|}}$ to $\frac{\delta^2}{\sqrt{2(|P|-2)}}$. The latter is strictly smaller because $\delta < 1$. Repeating this argument shows that we get the minimal value of ($\star$) if we assume that all parts in $\mathbf{P_2}$ are doubletons. In that case, ($\star$) becomes this:
	\begin{align*} \Orbit_{\Gamma_{\mathbf{P_2}}}(a) \geq \left(\frac{4\delta}{\sqrt{2}}\right)^{|[n] \setminus S_n|/2} \geq \sqrt{2}^{|[n] \setminus S_n|/2}.
	\end{align*}
	Now plugging in our initial assumption $|[n] \setminus S_n| \geq 8 \log n$, this expression is $\geq n^2$. This is a contradiction to the fact that $|A_n| \leq cn$, for some constant $c$.
\end{proof}
Now that we know that the set of singleton positions $S_n$ is almost $[n]$ itself, we proceed to construct our subset of the colour classes $J_n \subseteq \Pp_n$ such that the intersection of its supports already individualises all positions in $S_n$. Furthermore, we make sure that there are not too many ways how the supports of the colour classes in $J_n$ can be permuted simultaneously; this will show that $\bigcap_{C \in \Pp_n} \Stab(\Sp(C))$ cannot be too large, which is precisely our goal.
%TODO: replace f by h
\begin{lemma}
	\label{lem:LOG_levelIntersection}
	Let $f(n) \in o(n)$ such that (for large enough $n$) for all colour classes $C \in \Pp_n$, $|\Sp(C)| \leq f(n)$.\\
	For large enough $n$, there exists a subset $J_n \subseteq \Pp_n$ of colour classes with the following two properties:
	\begin{enumerate}[(1)]
		\item Every position in $S_n$ is in a singleton part of $\bigsqcap_{C \in J_n} \Sp(C)$.
		\item The following bound for the number of realisable simultaneous permutations of the supporting partitions holds:
		\begin{align*}
			| \{\bar{\sigma} \in \bigtimes_{C \in J_n} \Sym(\Sp(C)) &\mid \text{there is a } \pi \in \Sym_n \text{ that realises } \bar{\sigma} \} |\\
			&\leq 
			\left( f(n)!  \right)^{\lceil n / (f(n)-1)\rceil} \cdot 2^n
		\end{align*}
	\end{enumerate}
\end{lemma}
\begin{proof}
	We construct $J_n$ stepwise, starting with $J_n^0 := \emptyset$ and adding one new colour class $C_{j_i} \in \Pp_n$ in each step $i \geq 1$ in such a way that 
	\[
	\Bigl| \bigsqcap_{C \in J_n^{i-1}} \Sp(C) \sqcap \Sp(C_{j_i}) \Bigl| > \Bigl| \bigsqcap_{C \in J_n^{i-1}} \Sp(C) \Bigl| .
	\]
	Let $s$ be the number of construction steps needed, i.e.\ $J_n := J_n^s$ is such that property (1) of the lemma holds for this subset of $\Pp_n$. By definition of $S_n$, it is clear that such a subset exists because $\Pp_n$ itself satisfies property (1).\\
	For each construction step $i$, we let
	\[
	\Gamma_{i} := \{ \bar{\sigma} \in \bigtimes_{C \in J_n^i} \Sym(\Sp(C)) \mid \text{ there is a } \pi \in \Sym_n \text{ that realises } \bar{\sigma} \}.
	\]
	Furthermore, for each step $i$ we let $k_i$ be the increase in the number of parts in the intersection that is achieved in this step:
	\[
	k_i := | \bigsqcap_{C \in J_n^{i}} \Sp(C)   | - | \bigsqcap_{C \in J_n^{i-1}} \Sp(C)|.
	\]
	\\
	The main part of the proof consists in showing the following
	\begin{claim}
		For each step $i$, the size of $|\Gamma_{i}|$ is bounded by 
		\[
		|\Gamma_{i}| \leq \prod_{j = 1}^{i} (\min\{(k_j+1), f(n)\})!
		\]
	\end{claim}
	\noindent 
	\textit{Proof of claim.}
	Via induction on $i$. For $i = 1$, we have $k_1 = |\Sp(C_{j_1}) | \leq f(n)$. The group $\Gamma_1$ is a subgroup of $\Sym(\Sp(C_{j_1}))$, whose size is bounded by $|\Sp(C_{j_1})|!$. Therefore, the claim holds.\\
	\\
	For the inductive step, consider the step $i+1$ of the construction. Let $j_{i+1}$ be the colour class of the preorder that is added in this step. In order to bound the size of $\Gamma_{i+1}$, we consider for each $\bar{\sigma} \in \Gamma_i$ the following set:
	\[
	\Gamma_{i+1}^{\bar{\sigma}} :=  \{ \sigma \in \Sym(\Sp(C_{j_{i+1}})) \mid \text{there is a } \pi \in \Sym_n \text{ that realises } \sigma \text{ and } \bar{\sigma} \}. 
	\]
	We need to show that for each $\bar{\sigma} \in \Gamma_i$, it holds $|\Gamma_{i+1}^{\bar{\sigma}}| \leq (\min\{(k_{i+1}+1), f(n)\})!$.\\
	Since $|\Sp(C_{j_{i+1}})| \leq f(n)$, the bound $|\Gamma_{i+1}^{\bar{\sigma}}| \leq f(n)!$ is clear. It remains to show that for an arbitrary fixed $\bar{\sigma} \in \Gamma_i$, it holds $|\Gamma_{i+1}^{\bar{\sigma}}| \leq (k_{i+1}+1)!$.\\
	\\
	For a part $P \in \Sp(C_{j_{i+1}})$, let 
	\[
	\mathbf{Q}(P) :=  \{Q \in \bigsqcap_{C \in J_n^{i}} \Sp(C) \mid Q \cap P \neq \emptyset\}.
	\]
	Then we define an equivalence relation $\sim \subseteq \Sp(C_{j_{i+1}})^2$: For parts $P,P' \in \Sp(C_{j_{i+1}})$, we let
	\[
	P \sim P' \text{ iff } \mathbf{Q}(P) = \mathbf{Q}(P').
	\]
	Now we show how the equivalence classes of $\sim$ can be used to approximate the images of the parts in $\Sp(C_{j_{i+1}})$ under permutations in $\Gamma_{i+1}^{\bar{\sigma}}$:
	Every $\pi \in \Sym_n$ that realises any $\sigma \in  \Gamma_{i+1}^{\bar{\sigma}}$ also realises $\bar{\sigma} \in \Gamma_i$. Hence, by Lemma \ref{lem:LOG_centralLemmaForMultipleSets}, all such $\pi$ induce the same $\theta_{\bar{\sigma}} \in \Sym(\bigsqcap_{C\in J_n^i} \Sp(C))$. This means that for any $\sigma \in \Gamma_{i+1}^{\bar{\sigma}}$, and every part $P \in \Sp(C_{j_{i+1}})$, 
	\[
	\sigma(P) \subseteq  \bigcup_{Q \in \mathbf{Q}(P)} \theta_{\bar{\sigma}}(Q).
	\]
	
	The situation is visualised in the figure below. 
	
	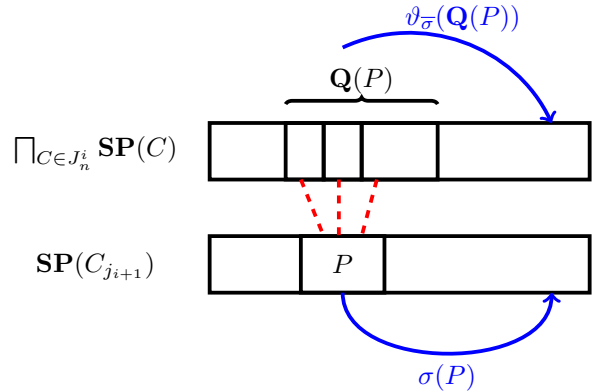
\begin{figure}[H]
		\centering
		\begin{tikzpicture}[line width = 0.5mm]
			\node at (-1.5,-0.4) {$\bigsqcap_{C \in J_n^i} \Sp(C)$};
			\node at (-1.5,-1.9) {$\Sp(C_{j_{i+1}})$};
			\draw (0,0) rectangle (5,-0.75);
			\draw (1,0) rectangle (1.5,-0.75);
			\draw (1.5,0) rectangle (2,-0.75);	
			\draw (2,0) rectangle (3,-0.75);	
			
			\draw (0,-1.5) rectangle (5,-2.25);
			\draw (1.2,-1.5) rectangle (2.3,-2.25) node[pos=.5] {$P$};
			
			\draw[decorate, decoration={brace}] (1,0.25) -- ( 3 ,0.25 ) node[pos=0.5, above] {$\mathbf{Q}(P)$};
			
			\draw[red,dashed] (1.2,-0.75) -- (1.5,-1.5);
			\draw[red,dashed] (1.7,-0.75) -- (1.7,-1.5);				
			\draw[red,dashed] (2.2,-0.75) -- (2,-1.5);				
			
			\draw (1.75,-2.25) edge[blue,->, bend right=90] node [below] {$\sigma(P)$} (4.5,-2.25);
			\draw (1.75,1) edge[blue,->, bend left=45] node [above] {$\theta_{\bar{\sigma}}(\mathbf{Q}(P))$} (4.5,0);
			%\draw[pattern = north west lines, pattern color = yellow]  (0,-2) rectangle (5,-2.75);	
			%\draw[preaction={fill = blue!50}, pattern = north west lines, pattern color = yellow] (1,-2) rectangle (3,-2.75);
			%\draw[preaction={fill = blue!50}, pattern = north west lines, pattern color = red] (3,-2) rectangle (5,-2.75);
		\end{tikzpicture}
		\caption{A part $P \in \Sp(C_{j_{i+1}})$ and the associated parts $\mathbf{Q}(P)$ in  $\bigsqcap_{C \in J_n^i} \Sp(C)$ that it intersects. Any $\sigma \in \Sym(\Sp(C_{j_{i+1}}))$ compatible with $\bar{\sigma}$ must map $P$ to the same position that $\mathbf{Q}(P)$ is mapped to by $\theta_{\bar{\sigma}}$.}
	\end{figure}	
	Therefore, we can fix any $\widehat{\sigma} \in  \Gamma_{i+1}^{\bar{\sigma}}$ and obtain:
	\[
	\{ \sigma(P) \mid \sigma \in \Gamma_{i+1}^{\bar{\sigma}} \} \subseteq [\widehat{\sigma}(P)]_\sim.
	\]
	Note that the class $[\widehat{\sigma}(P)]_\sim$ is independent of the choice of $\widehat{\sigma}$. Consequently, we can bound $|\Gamma_{i+1}^{\bar{\sigma}}|$ as follows: Let $m$ be the number of equivalence classes of $\sim$ and let $\ell_1,..,\ell_m$ denote the sizes of the respective classes. Then from our observations so far it follows:
	\begin{equation*}
		|\Gamma_{i+1}^{\bar{\sigma}}| \leq \prod_{t \in [m]} \ell_t! \tag{$\star$}
	\end{equation*}
	Next, we establish a relationship between the properties of $\sim$ and the number $k_{i+1}$:
	\begin{align*}
		k_{i+1} &=  \Bigl| \Sp(C_{j_{i+1}}) \sqcap \bigsqcap_{C \in J_n^{i}} \Sp(C)  \Bigl| - \Bigl| \bigsqcap_{C \in J_n^{i}} \Sp(C)  \Bigl|\\
		&= \sum_{Q \in \bigsqcap_{C \in J_n^{i}} \Sp(C)} (|\{P \in \Sp(C_{j_{i+1}}) \mid Q \in \textbf{Q}(P) \}| - 1)\\
		&\geq \sum_{\stackrel{[P]_\sim,}{P \in \Sp(C_{j_{i+1}})}} (| [P]_\sim | - 1)\\
		&=|\Sp(C_{j_{i+1}})| - m. 	
	\end{align*}
	The inequality in this chain requires some explanation:
	Fix a choice function $g$ that maps each equivalence class $[P]_\sim$ to a part $Q \in \mathbf{Q}(P)$. By definition of $\sim$, we have $g([P]_\sim) \in \mathbf{Q}(P')$ for every $P' \in [P]_\sim$. Hence, for every $Q \in \bigsqcap_{C \in J_n^{i}} \Sp(C)$, it holds:\\ $|\{P \in \Sp(C_{j_{i+1}}) \mid Q \in \textbf{Q}(P) \}| - 1 \geq \sum_{[P]_\sim \in g^{-1}(Q)} (|[P]_\sim| - 1)$.\\  
	We can sum up the result of these considerations like this:
	\begin{equation*}
		m \geq | \Sp(C_{j_{i+1}}) | - k_{i+1}.\tag{$\star\star$}
	\end{equation*} 
	This means that the relation $\sim$ has rather many equivalence classes if $k_{i+1}$ is small. Let us now finish the proof of the claim:\\ 
	
	We have already established the upper bound ($\star$) for $|\Gamma_{i+1}^{\bar{\sigma}}|$. Let $p \in [m]$ be such that $\ell_p \geq \ell_t$ for all $t \in [m]$. A consequence of ($\star \star$) is: $\ell_p \leq k_{i+1}+1$. It can be checked that the values $\ell_1,...,\ell_m$ that maximise the bound in ($\star$) and satisfy ($\star \star$) are such that $\ell_t = 1$ for all $t \neq p$: If there is some $t \neq p$ with $\ell_t \geq 2$, then decrease $\ell_t$ by one and increase $\ell_p$ by one. This does not change the number $m$ of equivalence classes, so it still satisfies ($\star \star$), but the value of $\prod_{t \in [m]} \ell_t!$ is strictly increased.\\
	Therefore, ($\star$) becomes:
	\[
	|\Gamma_{i+1}^{\bar{\sigma}}| \leq \ell_p! \leq (k_{i+1}+1)!
	\]
	This concludes the proof of the claim. \qedsymbol\\
	\\
	Hence, in order to finish the proof of the lemma, we have to bound
	\[
	|\Gamma_{s}| \leq \prod_{i = 1}^{s} (\min\{(k_i+1), f(n)\})!
	\]
	from above (recall that $s$ is the number of steps needed to construct $J_n$ satisfying property (1)). We know that $\sum_{i=1}^{s} k_i$ is some fixed value $\leq n$. The values of the above product and of the sum solely depend on the sequence $(k_i)_{i \in [s]}$. Now we can make a ``redistribute-weight argument'' again: Let $k_j$ be such that $k_j + 1 < f(n)$, and such that there is a $k_i \leq k_j$ with $k_i > 1$. We can decrease $k_i$ by one and increase $k_j$ by one. This does not change the value of the sum, and the value of the product of factorials can only get larger (as $k_j+1$ does not exceed $f(n)$). If we iterate this process, we see that the value of the product is maximised for a sequence $(k_i)_{i \in [s]}$, where every $k_i$ is either $1$ or $k_i = f(n)-1$ (and there may be exactly one $k_i$ with $1 < k_i < f(n) - 1$).\\
	For such a sequence of $k_i$, the value of the product is at most
	\[
	|\Gamma_{s}| \leq \prod_{i = 1}^{s} (\min\{(k_i+1), f(n)\})! \leq f(n)!^{ n /(f(n)-1) } \cdot 2^n
	\]
\end{proof}

\begin{corollary}
	\label{cor:LOG_finalCorollaryCase1}
	Assume the following three conditions hold:
	\begin{enumerate}
		\item Every $v \in \{0,1\}^n$ occurs in at least one of the colour classes $C$ in $\Pp_n$.
		\item The function $\max_{C \in \Pp_n}{|C|}$ is in $\Oo(n)$.
		\item $|\Sp(B_n)| \in o(n)$.
	\end{enumerate}
	Then, for sufficiently large $n$:
	\[
	|\Stab(\Pp_n)| \leq \left( f(n)!  \right)^{n / (f(n)-1)} \cdot 2^n \cdot (8 \log n)!
	\]
\end{corollary}
\begin{proof}	
	Consider the set of colour classes $J_n \subseteq \Pp_n$ that exists by Lemma \ref{lem:LOG_levelIntersection}. Since every $\pi \in \Stab(\Pp_n)$ fixes every $C \in \Pp_n$, by Lemma \ref{lem:sandwichLemma}, it induces a tuple of permutations on the supporting partitions $\bar{\sigma} \in \bigtimes_{C \in \Pp_n} \Sym(\Sp(C))$. In particular it also induces a $\bar{\sigma} \in \bigtimes_{C \in J_n} \Sym(\Sp(C))$. By Lemma \ref{lem:LOG_levelIntersection}, there are at most $\left( f(n)!  \right)^{ n / (f(n)-1)} \cdot 2^n$ possibilities for such a $\bar{\sigma}$. Furthermore, each such $\bar{\sigma}$ can be realised by at most $(8 \log n)!$ distinct permutations $\pi \in \Stab(\Pp_n)$: Due to Lemma \ref{lem:LOG_centralLemmaForMultipleSets} and property (1) of $J_n$ (see Lemma \ref{lem:LOG_levelIntersection}), every $\pi$ realising $\bar{\sigma}$ permutes the positions in $S_n$ in the same way, and according to Lemma \ref{lem:nonSingletonsUpperBound}, there remain at most $8 \log n$ positions which may be permuted arbitrarily by $\pi$.
\end{proof}

With this, we prove our final lemma of this subsection, which is just a more precise formulation of Lemma \ref{lem:resultCaseSublinear}.

\begin{lemma}
	Assume again the three conditions from Corollary \ref{cor:LOG_finalCorollaryCase1}.
	Then for any $k \in \bbN$, the limit 
	\[
	\lim_{n\to\infty} \frac{|\Orbit(\Pp_n)|}{2^{kn}} = \lim_{n\to\infty} \frac{n!}{|\Stab(\Pp_n)| \cdot 2^{kn}} 
	\]
	does not exist. That is to say, the orbit of $\Pp_n$ w.r.t.\ the action of $\Sym_n$ on $\{0,1\}^n$ grows super-polynomially in $2^n$.
\end{lemma}	
\begin{proof}
	Plugging in the stabiliser bound from Corollary \ref{cor:LOG_finalCorollaryCase1}, we get for the fraction that we are taking the limit of:
	\[
	\frac{n!}{|\Stab(\Pp_n)| \cdot 2^{kn}} \geq \frac{n!}{\left( f(n)! \right)^{ n/(f(n)-1) } \cdot (8 \log n)! \cdot 2^{(k+1)n}}
	\]
	According to Stirling's Formula, factorials can be approximated as follows: 
	\[
	0.5 \cdot n! \leq \sqrt{2\pi n} \cdot \left(\frac{n}{e}\right)^n \leq n!
	\] 
	\noindent
	With this, we obtain the following chain of expressions, where we use the abbreviation $h(n) := \frac{f(n)+0.5}{f(n)-1}$.
	\begin{align*}
		&\frac{n!}{\left( f(n)! \right)^{ n/(f(n)-1)} \cdot (8 \log n)! \cdot 2^{(k+1)n}}\\
		\geq
		&\frac{1}{2} \cdot \sqrt{\frac{n}{8\log n}} \cdot \frac{1}{(8\pi f(n))^{0.5n/(f(n)-1)}} \cdot \left(\frac{n}{2^{k+1}e}\right)^{n} \cdot  \left(\frac{e}{8 \log n}\right)^{8 \log n} \cdot \left(\frac{e}{f(n)}\right)^{f(n) \cdot (n/(f(n)-1))}\\
		\geq & \frac{1}{2} \cdot \sqrt{\frac{n}{8\log n}} \cdot 
		\left(\frac{n}{(8\pi)^{0.5/(f(n)-1)} \cdot 2^{k+1}e \cdot f(n)^{h(n)}} \right)^{n} \cdot  \left(\frac{e}{8 \log n}\right)^{8 \log n}\\
		= & \frac{1}{2} \cdot\sqrt{\frac{n}{8\log n}} \cdot 
		\left(\frac{ n^{(n- 8\log \log n)/n} }{(8\pi)^{0.5/(f(n)-1)} \cdot 2^{k+1}e \cdot f(n)^{h(n)}} \right)^{n} \cdot \left(\frac{\log n \cdot e}{8 \log n}\right)^{8 \log n}\\
		= & \frac{1}{2} \cdot\sqrt{\frac{n}{8\log n}} \cdot 
		\left(\frac{ n^{(n- 8\log \log n)/n} }{(8\pi)^{0.5/(f(n)-1)} \cdot 2^{k+1}e \cdot f(n)^{h(n)}} \right)^{n} \cdot \left(\frac{e}{8}\right)^{8 \log n}\\
	\end{align*}
	In this calculation we used that $n^{8 \log \log n} = (\log n)^{8 \log n}.$ Now we need the following claims:\\
	\begin{claim} 
		$
		f(n)^{h(n)} \in \Theta(f(n)).
		$
	\end{claim}
	\begin{proof}
		\begin{align*}
			\lim_{n\to\infty} \frac{f(n)^{h(n)}}{f(n)} = \lim_{n\to\infty} f(n)^{\frac{1.5}{f(n)-1}} = \lim_{n\to\infty} \exp\Bigl(\frac{1.5}{f(n)-1} \cdot \ln(f(n))\Bigr) = 1.
		\end{align*} 	
	\end{proof}

	\begin{claim}
		$
		n^{(n - 8 \log \log n)/n} \in \Theta(n).
		$
	\end{claim}
	\begin{proof}
		\begin{align*}
			\lim_{n\to\infty} \frac{n^{(n - 8 \log \log n)/n}}{n} = \lim_{n\to\infty} n^{(-8\log \log n)/n} = \lim_{n\to\infty} \exp\Bigl(\frac{-8\log \log n}{n} \cdot \ln(n)\Bigr) = 1.
		\end{align*}  
	\end{proof}
	\noindent 
	From these claims and the fact that $f(n) \in o(n)$ it follows that there is no constant $c$ such that
	\[
	\frac{n^{(n - 8 \log \log n)/n}}{f(n)^{h(n)}} \leq c,
	\]
	for all $n$. We can conclude that for large enough $n$,
	\[
	\left(\frac{n^{(n - 8\log \log n)/n}}{(8\pi)^{0.5/(f(n)-1)} \cdot 2^{k+1}e \cdot f(n)^{h(n)}}\right)^{n} \geq 2^n.
	\]
	Furthermore,
	\[
	\left(\frac{e}{8}\right)^{8 \log n} \geq 	\left(\frac{n^8}{n^{24}}\right) = n^{-16}.
	\]
	So, the dominating factor is $2^n$, which means that the limit of the whole product does not exist.
\end{proof}
\noindent With this, we have proven the super-polynomial orbit theorem for preorders with log-sized colour classes in case that the supporting partitions of the colour classes have at most sublinearly many parts. We now move on to the remaining case.\\

\subsection{The case of linearly-sized supports} 
This subsection is dedicated to proving the result for the case that $|\Sp(B_n)| \in \Theta(n)$. Recall that $B_n$ denotes the colour class $C \in \Pp_n$ whose supporting partition has the most parts. 
\begin{lemma}
	\label{lem:LOG_caseLinearSupport}
	Assume that the following conditions hold for $B_n$:
	\begin{enumerate}
		\item $|B_n| \in \Oo(n)$.
		\item $|\SpC(B_n)| \in \Theta(n)$
	\end{enumerate}
	Then the orbit size of $B_n$ (and therefore also of $\Pp_n$) w.r.t.\ the action of $\Sym_n$ on $\{0,1\}^n$ grows faster than any polynomial in $2^n$.
\end{lemma}
Proving this lemma requires several steps, and again, a case distinction. The relevant measure here is the number of \emph{singleton} parts in $\Sp(B_n)$. Firstly, we show that if the number of singleton parts in $\SpC(B_n)$ grows sublinearly in $n$, while the total number of parts $|\SpC(B_n)|$ is linear, the stabiliser of $\SpC(B_n)$ is small enough such that \autoref{lem:LOG_caseLinearSupport} is true. This is a relatively straightforward calculation.\\
\\
\textbf{Subcase 1: Sublinear number of singleton parts}\\
Let us begin with the easier case, where the number of singleton parts in $\Sp(B_n)$ grows sublinearly. We denote by $S_n \subseteq [n]$ the set of positions that are in singleton parts, i.e.\
%TODO: mention that this is different from S_n in the previous setion
\[
S_n := \{k \in [n] \mid \{k\} \in \Sp(B_n)  \}.
\]
Thus, the meaning of $S_n$ is now slightly different as in the previous section.
The size of $\Stab(\Sp(B_n))$ can be bounded as follows:
\begin{lemma}
	\label{lem:LOG_stabBoundSingletons}
	Let $s_n := |S_n|$, and $t_n := |\Sp(B_n)| - s_n$.
	\[
	|\Stab(\Sp(B_n))| \leq s_n! \cdot t_n! \cdot \left((n/t_n)!)\right)^{t_n}.
	\]  
\end{lemma}
\begin{proof}
	Recall that $\Stab(\Sp(B_n))$ is the setwise stabiliser of the support, so it also involves permutations that map parts to other parts. The factors $s_n!$ and $t_n!$ account for these possible permutations of the parts: All the singleton parts of $\Sp(B_n)$ can be mapped to each other, and every non-singleton part can at most be mapped to every other non-singleton part. The factor $\left((n/t_n)!)\right)^{t_n}$ is an upper bound for the number of possible permutations within all non-singleton parts, as every non-singleton part can have size at most $(n/t_n)$. 
\end{proof}

\begin{corollary}
	\label{cor:LOG_stabBoundSingletonsConcrete}
	Let $f(n) \in o(n)$ be a function such that $s_n \leq f(n)$ and assume $|\Sp(B_n)| \in \Theta(n)$, i.e.\ there exists a constant $0 < c \leq 1$, such that for all large enough $n$, $|\Sp(B_n)| \geq c \cdot n$. Then, for all large enough $n$, the following bound holds:
	\[
	|\Stab(\Sp(B_n))| \leq f(n)! \cdot (n/2)! \cdot  \left(d! \right)^{(n/2)}.
	\]
	Here, $d$ is some positive constant.
\end{corollary}
\begin{proof}
	We plug in the right values for $s_n$ and $t_n = |\Sp(B_n)| - n$ into Lemma \ref{lem:LOG_stabBoundSingletons}. We have $s_n \leq f(n)$ by assumption. Further, using $|\Sp(B_n)| \geq c \cdot n$, and the fact that every non-singleton part consists of at least two elements, we can bound $t_n$ as follows:
	\[
	c \cdot n - f(n) \leq t_n \leq \frac{n}{2}.
	\]
	Now plug in the upper and lower bounds for $t_n$ and $s_n$ in the right places in the bound from Lemma \ref{lem:LOG_stabBoundSingletons}. Note that the expression $\frac{1}{c - f(n)/n}$ that occurs in this bound can be upper-bounded by some $\frac{1}{c - \epsilon}$ (for large enough $n$) because $f(n) \in o(n)$. Then set $d := \frac{1}{c - \epsilon}$.
\end{proof}

\begin{lemma}
	\label{lem:LOG_orbitCase11}
	Let $f(n) \in o(n)$ be a function such that $s_n \leq f(n)$ and assume $|\Sp(B_n)| \in \Theta(n)$. Then for any $k \in \bbN$, the limit 
	\[
	\lim_{n\to\infty} \frac{|\Orbit(B_n)|}{2^{kn}} \geq \lim_{n\to\infty} \frac{n!}{|\Stab(\Sp(B_n))| \cdot 2^{kn}} 
	\]
	does not exist. That is to say, the orbit of $B_n$ w.r.t.\ the action of $\Sym_n$ on $\{0,1\}^n$ grows super-polynomially in $2^n$.
\end{lemma}
\begin{proof}
	We bound $|\Stab(\Sp(B_n))|$ according to the preceding corollary. Factorials can be approximated by the Stirling Formula (in fact, the approximation is much closer to $n!$ than what we state here, but this is sufficient and makes the calculations nicer):
	\[
	0.5 \cdot n! \leq \sqrt{2\pi n} \cdot \left(\frac{n}{e}\right)^n \leq n!
	\] 
	Using the upper bound for $|\Stab(\SpC(B_n))|$ from Corollary \ref{cor:LOG_stabBoundSingletonsConcrete}, and the Stirling Formula for the factorials $n!, f(n)!$, and $(n/2)!$, we get:
	\begin{align*}
		\frac{n!}{|\Stab(\Sp(B_n))| \cdot 2^{kn}} &\geq 
		\frac{1}{4} \cdot \sqrt{\frac{n}{2\pi \cdot f(n) \cdot (n/2)}} \cdot \left(\frac{n}{e}\right)^n \\
		&\cdot \left(\frac{e}{f(n)}\right)^{f(n)} \cdot \left(\frac{e}{n/2}\right)^{(n/2)} \cdot \frac{1}{(d!)^{(n/2)} \cdot 2^{kn+2}}\\
		&\geq \frac{1}{4} \cdot \frac{1}{\sqrt{\pi \cdot f(n)}} \cdot \left(\frac{n}{e} \right)^{(n/2)-f(n)} \cdot \frac{1}{(d!)^{(n/2)} \cdot 2^{kn}}.
	\end{align*}
	For the last inequality, we cancelled some of the factors $\frac{n}{e}$ with the other factors, leaving a factor $> 1$ each time.\\
	The factor $\frac{1}{(d!)^{(n/2)} \cdot 2^{kn}}$ can be written as $\epsilon^{(n/2)}$, for a small enough constant $\epsilon > 0$. For large enough $n$, it holds that $n-2f(n) \geq n/2$, since $f(n) \in o(n)$. Hence, we get in total:
	\begin{align*}
		\lim_{n\to\infty} \frac{n!}{|\Stab(\SpC(B_n))| \cdot 2^{kn}} &\geq \lim_{n\to\infty} \frac{1}{4\sqrt{\pi \cdot f(n)}} \cdot \left(\frac{\epsilon^2 \cdot n}{e} \right)^{(n/2)-f(n)}
	\end{align*}
	Again, observing that $f(n) \in o(n)$, it can be seen that this limit does not exist.
\end{proof}
\noindent
This proves Lemma \ref{lem:LOG_caseLinearSupport} under the assumption that $S_n \in o(n)$. Now we deal with the remaining case.\\
\\
\textbf{Subcase 2: Linear number of singleton parts}\\
As already mentioned, this case requires more effort because we cannot solve it by only counting the number of permutations that stabilise $\Sp(B_n)$. Instead, we have to relate $\Sp(B_n)$ to the set $B_n$. Let $\Sym(B_n)$ be the group of all permutations of the strings in $B_n$. For $\pi \in \Sym_n$ and $\sigma \in \Sym(B_n)$, we say that $\pi$ \emph{realises} or \emph{induces} $\sigma$, if $\pi(b) = \sigma(b)$ for every $b \in B_n$.\\
Each $\pi \in \Stab(B_n) \leq \Sym_n$ that permutes the positions of the strings in $B_n$ induces a unique permutation in $\Sym(B_n)$. Conversely, each $\sigma \in \Sym(B_n)$ can be realised by multiple distinct $\pi \in \Stab(B_n)$, but not by too many: Roughly speaking, we will see that the number of distinct realisations of a $\sigma \in \Sym(B_n)$ is related to $|\StabP(\Sp(B_n))|$ (which is small if there are many singleton parts). Therefore, the aim is to show that only a bounded number of $\sigma \in \Sym(B_n)$ can be realised by a permutation $\pi \in \Stab(B_n)$ at all, and that each such $\sigma$ only has a small number of realisations. In total, this yields a bound on $|\Stab(B_n)|$.\\

Now let us go into the details: For a set $A \subseteq \{0,1\}^n$, we write $\bigsqcap A$ for $\bigsqcap_{a \in A} \Sp(a)$. Note that  $\Sp(a)$ is just the partition of $[n]$ into the positions where there are zeros and ones, respectively, in $a$. It can be seen that $\bigsqcap A$ is a supporting partition for $A$, which is finer or identical to its coarsest supporting partition.\\
First, we show: If we partially specify a $\sigma \in \Sym(B_n)$ by only fixing its behaviour on a subset $A \subseteq B_n$, then any $\sigma \in \Sym(B_n)$ compliant with the specification can only be realised by permutations $\pi \in \Sym_n$ which respect the parts of $\bigsqcap A$ in some way. Roughly speaking, if $\bigsqcap A$ consists of many small parts, then each such $\sigma$ will only have a small number of realisations in $\Sym_n$.\\
The second step is then to choose a suitable set $A_n \subseteq B_n$ such that fixing the behaviour of $\sigma \in \Sym(B_n)$ on $A_n$ indeed only admits a small number of realisations of $\sigma$, and such that $A_n$ is sufficiently small to admit only few possibilities how $\sigma$ can behave on $A_n$ (as $|B_n| \in \Oo(n)$, there are $\approx n^{|A_n|}$ such possibilities).\\

First of all, we show how to bound the number of possible realisations of any $\sigma \in \Sym(B_n)$ if $\sigma$ is fixed on some subset $A \subseteq B_n$. The next lemma is of a similar flavour as Lemma \ref{lem:LOG_centralLemmaForMultipleSets}. 
\begin{lemma}
	\label{lem:LOG_centralLemma}
	Let $B \subseteq \{0,1\}^n, A \subseteq B$. Let an injective mapping $p : A \lra B$ be given. Write $\bigsqcap A := \bigsqcap_{a \in A} \Sp(a)$.\\
	\\
	There is an assignment of positions to parts $Q_p : [n] \lra \bigsqcap A$ with the property that $|Q_p^{-1}(P)| = |P|$ for every $P \in \bigsqcap A$, and such that:\\
	Every $\pi \in \Sym_n$ realising any $\sigma \in \Sym(B)$ with $\sigma^{-1}(a) = p(a)$ for all $a \in A$ satisfies: $\pi(k) \in Q_p(k)$ for all $k \in [n]$ (it may be that such a $\pi$ does not exist). 
\end{lemma}
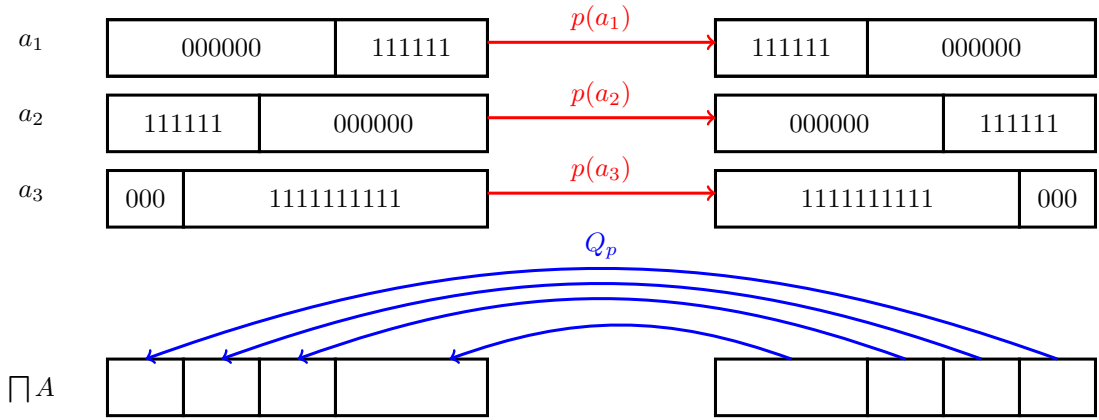
\begin{figure}[H]
	\centering
	\begin{tikzpicture}[line width = 0.4mm]
		\node at (-1,-0.3) {$a_1$};
		\node at (-1,-1.3) {$a_2$};
		\node at (-1,-2.3) {$a_3$};
		\node at (-1,-4.9) {$\bigsqcap A$};	
		\draw (0,0) rectangle (5,-0.75);
		\draw (0,0) rectangle (3,-0.75) node[pos=.5] {$000000$};
		\draw (3,0) rectangle (5,-0.75) node[pos=.5] {$111111$};
		
		\draw (0,-1) rectangle (2,-1.75) node[pos=.5] {$111111$};
		\draw (2,-1) rectangle (5,-1.75) node[pos=.5] {$000000$};
		
		\draw (0,-2) rectangle (1,-2.75) node[pos=.5] {$000$};	
		\draw (1,-2) rectangle (5,-2.75) node[pos=.5] {$1111111111$};	
		
		\draw (0,-3.5-1) rectangle (1,-4.25-1); %node[pos=.5] {$111111$};
		\draw (1,-3.5-1) rectangle (2,-4.25-1); %node[pos=.5] {$111111$};
		\draw (2,-3.5-1) rectangle (3,-4.25-1); %node[pos=.5] {$111111$};
		\draw (3,-3.5-1) rectangle (5,-4.25-1); %node[pos=.5] {$111111$};

		\draw (8,0) rectangle (5+8,-0.75);
		\draw (8,0) rectangle (2+8,-0.75) node[pos=.5] {$111111$};
		\draw (2+8,0) rectangle (5+8,-0.75) node[pos=.5] {$000000$};
		
		\draw (0+8,-1) rectangle (3+8,-1.75) node[pos=.5] {$000000$};
		\draw (3+8,-1) rectangle (5+8,-1.75) node[pos=.5] {$111111$};
		
		\draw (0+8,-2) rectangle (4+8,-2.75) node[pos=.5] {$1111111111$};	
		\draw (4+8,-2) rectangle (5+8,-2.75) node[pos=.5] {$000$};	
		
		\draw (8,-3.5-1) rectangle (2+8,-4.25-1); %node[pos=.5] {$111111$};
		\draw (2+8,-3.5-1) rectangle (3+8,-4.25-1); %node[pos=.5] {$111111$};
		\draw (3+8,-3.5-1) rectangle (4+8,-4.25-1); %node[pos=.5] {$111111$};
		\draw (4+8,-3.5-1) rectangle (5+8,-4.25-1); %node[pos=.5] {$111111$};

		\draw[red,->] (5,-0.3) -- node[above] {$p(a_1)$}  (8,-0.3);
		\draw[red,->] (5,-1.3) -- node[above] {$p(a_2)$}  (8,-1.3);
		\draw[red,->] (5,-2.3) -- node[above] {$p(a_3)$}  (8,-2.3);
		
		\draw (9,-3.5-1) edge[blue,->, bend right=20] node [above] {} (4.5,-3.5-1);
		\draw (10.5,-3.5-1) edge[blue,->, bend right=20] node [above] {} (2.5,-3.5-1);
		\draw (11.5,-3.5-1) edge[blue,->, bend right=20] node [above] {} (1.5,-3.5-1);
		\draw (12.5,-3.5-1) edge[blue,->, bend right=20] node [above] {$Q_p$} (0.5,-3.5-1);
		
	\end{tikzpicture}
	\caption{An example with $A = \{a_1,a_2,a_3\} \subseteq B$ and a mapping $p$ that specifies an image for each string in $A$. Every $\sigma \in \Sym(B)$ that acts as the inverse of $p$ on $A$ can only be realised by a $\pi \in \Sym_n$ that complies with $Q_p$.}
\end{figure}

\begin{proof}
	Let $\sigma \in \Sym(B)$ be such that $\sigma(a) = p(a)$ for each $a \in A$ and assume $\sigma$ is indeed a permutation of the strings in $B$ that can be realised by at least one $\pi \in \Sym_n$. We show the statement via induction on $|A|$.\\
	
	If $|A| = 1$, then it is determined that $\sigma(p(a)) = a$, for the only string $a \in A$ (and $a$ and $p(a)$ must have the same Hamming-weight). In order to map $p(a)$ to $a$, the ones and zeros must be mapped correctly. Indeed, this automatically yields the desired assignment $Q_p$: Let $P_0, P_1 \subseteq [n]$ be the positions where there are $0$s and $1$s, respectively, in $a$. Then $Q_p(k) = P_0$, if and only if $p(a)_k = 0$.\\ 
	
	For the inductive step, assume the statement holds for $|A| \leq m$. Consider now $A$ with $|A| = m+1$. Take any $m$-element subset $A' \subset A$. Any $\pi \in \Sym_n$ realising $\sigma$ must in particular induce the correct preimages on the elements of $A'$, and therefore, $\pi$ has to respect the assignment $Q'_p$ given by the induction hypothesis. That is, for each $k \in [n]$, it holds $\pi(k) \in Q'_p(k)$, where $Q'_p(k)$ is a part of  $\bigsqcap A'$. Let $a \in A \setminus A'$ be the unique string not contained in $A'$. Let $\mathbf{0}(a) := \{k \in [n] \mid a_k = 0 \}$, and $\mathbf{1}(a) := \{k \in [n] \mid a_k = 1 \}$.\\
	
	We define the desired assignment $Q_p : [n] \lra \bigsqcap A$ as follows:
	\begin{align*}
		Q_p(k) &:= \begin{cases}
			Q'_p(k) \cap \mathbf{0}(a) & \text{, if } p(a)_k = 0\\ 
			Q'_p(k) \cap \mathbf{1}(a) & \text{, if } p(a)_k = 1
		\end{cases} 
	\end{align*}
	It is easily seen that the range of $Q_p$ is indeed $\bigsqcap A$, since the range of $Q'_p$ is $\bigsqcap A'$, and we have $\bigsqcap A = \bigsqcap A' \sqcap \Sp(a) =  \bigsqcap A' \sqcap \{ \mathbf{0}(a),\mathbf{1}(a) \}$.\\
	Furthermore, it is clear that any $\pi \in \Sym_n$ realising $\sigma$ must map every $k \in [n]$ to a position within $Q_p(k)$: By the induction hypothesis, $\pi$ must map every $k$ to a position within $Q'_p(k)$, and since $\pi(p(a)) = a$, the zeros and ones in $p(a)$ must be moved to zeros and ones in $a$. 	
\end{proof}

\noindent We will mainly need this lemma for the restriction of the parts in $\bigsqcap A$ to the positions $S_n$ which are singletons in $\Sp(B_n)$. Therefore, we state the following important corollary:

\begin{corollary}
	\label{cor:LOG_centralCorollary}
	Let $A \subseteq B_n$ be arbitrary, and let an injective mapping $p : A \lra B$ be given. Then every $\pi \in \Sym_n$ that realises a $\sigma \in \Sym(B_n)$ with $\sigma^{-1}(a) = p(a)$ for all $a \in A$ satisfies:
	\[
	\pi^{-1}(P \cap S_n) = Q^{-1}_p(P) \cap S_n \text{ for all } P \in \bigsqcap A,
	\]
	where $Q_p : [n] \lra \bigsqcap A$ is the assignment that exists by the preceding lemma.
\end{corollary}
\begin{proof}
	Lemma \ref{lem:LOG_centralLemma} says that $\pi^{-1}(P) = Q^{-1}_p(P)$. Since $\pi$ realises a permutation in $\Sym(B_n)$, $\pi \in \Stab(B_n)$. Hence, by Lemma \ref{lem:sandwichLemma}, $\pi \in \Stab(\Sp(B_n))$. This means that $\pi(S_n) = S_n$, as singleton parts can only be mapped to singleton parts. Consequently, it must be the case that $\pi^{-1}(P \cap S_n) = Q^{-1}_p(P) \cap S_n$.  
\end{proof}

We move on to the more involved step of the proof. Recall that $S_n \subseteq [n]$ is the set of positions that are in singleton parts in $\Sp(B_n)$. As the next simple lemma shows, the intersection over all strings in $B_n$ yields a refinement of the coarsest supporting partition $\Sp(B_n)$, so in particular, the positions in $S_n$ are also in singleton parts in $\bigsqcap B$.

\begin{lemma}
	\label{lem:LOG_intersectionSupport}
	Let $B \subseteq \{0,1\}^n$. The partition $\bigsqcap B$ is a supporting partition for $B$.
\end{lemma}
\begin{proof}
	By the definition of the intersection, every string $b \in B_n$ is constant zero or constant one on every part $P \in \bigsqcap B$. Hence, $\StabP(\bigsqcap B) \subseteq \Stab(B_n)$. This is the definition of a supporting partition.	
\end{proof}

With this in mind, we can now select a subset $A_n \subseteq B_n$ such that the partition $\bigsqcap A_n$ is sufficiently fine on the positions in $S_n$ (such a subset must exist because $B_n$ itself fulfils this condition). If we take a suitable definition of what it means to be sufficiently fine on $S_n$, we can even guarantee that $A_n$ is considerably smaller than $B_n$. Later it will become clear that $A_n$ is chosen in such a way that if we partially specify a $\sigma \in \Sym(B_n)$ on $A_n$ (as in Lemma \ref{lem:LOG_centralLemma}), then there are only few permutations in $\Sym_n$ that realise this $\sigma$. The small size of $A_n$ then implies that not too many permutations in $\Sym(B_n)$ can at all be realised by permutations in $\Stab(B_n)$.

\begin{lemma}
	\label{lem:LOG_existenceOfSubsetA}
	There exists a subset $A_n \subseteq B_n$ of size $|A_n| \leq \frac{|S_n|}{2}$ such that for each part $P \in \bigsqcap A_n$, one of the following two statements is true:
	\begin{enumerate}
		\item $|P \cap S_n| \leq 2$; \textbf{or:}
		\item $|P \cap S_n| > 2$ and for every $b \in B_n \setminus A_n$, one of these two conditions holds:
		\begin{itemize}
			\item $b$ is constant on $P \cap S_n$; \textbf{or}
			\item $b[P \cap S_n]$ is \emph{imbalanced} and, for every $P' \in \bigsqcap A_n$ with $P' \neq P$, $|P' \cap S_n| > 2$, $b$ is \emph{constant} on $P' \cap S_n$.
		\end{itemize} 	
	\end{enumerate}
	By $b[P \cap S_n]$ we mean the substring of $b$ at the positions in $P \cap S_n$, and being \emph{imbalanced} means that $b[P \cap S_n]$ contains exactly one $0$ and there is a $1$ at all other positions, or vice versa (exactly one $1$ and the rest $0$). Being \emph{constant} means that the string has only zeros or only ones at the respective positions.
\end{lemma}

\begin{figure}[H]
	\centering
	\begin{tikzpicture}[line width = 0.5mm]
		\node at (-1,-0.4) {$\bigsqcap A_n$};
		\node at (-1.2,-1.9) {$b \in B_n \setminus A_n$};
		%	\node at (-1.9,-2.3) {$\Sp(A_1) \sqcap \Sp(A_2)$};
		%\draw (0,0) rectangle (5,-0.75);
		%\draw[fill = blue!50] (1,0) rectangle (5,-0.75);
		
		%	\draw[fill = yellow!30] (0,-1) rectangle (3,-1.75);
		%	\draw[fill = red!70] (3,-1) rectangle (5,-1.75);
		\draw[decorate, decoration={brace}] (0,0.2) -- (4.5 ,0.2 ) node[pos=0.5, above] {$S_n$};
		
		\draw[fill=green!20]  (0,0) rectangle (4.5,-0.75);	
		\draw[fill=green!20]  (0,-1.5) rectangle (4.5,-2.25);	
		
		\draw  (0,0) rectangle (0.2,-0.75);	
		\draw  (0.2,0) rectangle (0.4,-0.75);	
		\draw  (0.4,0) rectangle (0.6,-0.75);	
		\draw  (0.6,0) rectangle (0.8,-0.75);	
		\draw  (0.8,0) rectangle (1.0,-0.75);	
		\draw  (1.0,0) rectangle (1.2,-0.75);
		\draw  (1.2,0) rectangle (1.4,-0.75);	
		\draw  (1.4,0) rectangle (1.6,-0.75);	
		\draw  (1.6,0) rectangle (1.8,-0.75);	
		\draw  (1.8,0) rectangle (2.0,-0.75);	
		\draw  (2.0,0) rectangle (2.2,-0.75);	
		\draw  (2.2,0) rectangle (2.4,-0.75);	
		\draw  (2.4,0) rectangle (2.6,-0.75);	
		\draw (2.6,0)  rectangle (3.7,-0.75);
		\draw (3.7,0)  rectangle (4.5,-0.75);
		\draw (4.5,0)  rectangle (5,-0.75);

		\draw (0,-1.5) rectangle (2.6,-2.25) node[pos=.5] {$?$};;		
		\draw (2.6,-1.5) rectangle (3.7,-2.25) node[pos=.5] {$0001$};
		\draw (3.7,-1.5) rectangle (4.5,-2.25) node[pos=.5] {$111$};
		\draw (4.5,-1.5) rectangle (5,-2.25) node[pos=.5] {$?$};
	\end{tikzpicture}
	\caption{The set $A_n \subseteq B_n$ is chosen such that on the large parts of $\bigsqcap A_n$ in $S_n$ (the green area), every $b \in B_n \setminus A_n$ is constant/imbalanced. Elsewhere, $b$ may be arbitrary.}
\end{figure}
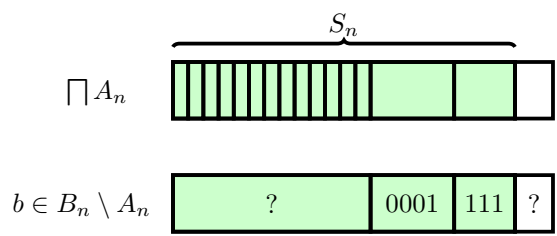

\begin{proof}
	We construct $A_n$ stepwise, starting with $A_n^0 := \emptyset$, and adding one string $a_i \in B_n$ in step $i$.\\
	For step $i+1$ of the construction, assume we have constructed $A_n^i$. For $k \in [n]$, we write $P_i(k)$ for the part of $\bigsqcap A_n^i$ that $k$ is in. Now we let
	\[
	K_i := \{k \in S_n \mid |P_i(k) \cap S_n| > 2 \}.
	\]
	This is the set of positions whose parts need to be refined more. If $K_i = \emptyset$, then the construction is finished because all parts of $\bigsqcap A_n^i$ satisfy condition 1 of the lemma. So assume $K_i \neq \emptyset$.\\
	\\
	By Lemma \ref{lem:LOG_intersectionSupport}, $\bigsqcap B_n$ is a supporting partition for $B_n$ and therefore at most as coarse as $\Sp(B_n)$. Hence, all positions in $S_n$ are in singleton parts of $\bigsqcap B_n$.\\
	We conclude that for all $k \in K_i$, there must be a string $b \in B_n \setminus A_n^i$ that can be added to $A_n^i$ in order to make $P_i(k) \cap S_n$ smaller when it is intersected with $\Sp(b)$. In fact, there may be several such strings $b$ that we could choose to add in this step of the construction. So let 
	\[
	C_k := \{ b \in B_n \setminus A_n^i \mid b \text{ is non-constant on } P_i(k) \cap S_n \}
	\] 
	be the non-empty set of such candidate strings. We restrict our candidate set further:
	\begin{align*}
		\widehat{C}_k := \{ b \in C_k &\mid \text{there are two distinct parts } P, P' \in \bigsqcap A_n^i\\ &\text{ s.t. } b \text{ is non-constant on } P \cap S_n \text{ and } P' \cap S_n, \text{ and} \\
		& |P \cap S_n| > 2 \text{ and } |P' \cap S_n| > 2\}\\
		\cup \{b \in C_k &\mid b[P_i(k) \cap S_n] \text{ is not imbalanced}\}.
	\end{align*}
	We pick our next string $a_{i+1}$ that is added in this step of the construction from one of the sets $\widehat{C}_k$, where $k$ ranges over all positions in $K_i$. If $\widehat{C}_k = \emptyset$ for all these $k$, then $A_n^i$ is already the desired set $A_n$ because it satisfies the conditions of the lemma.\\
	Otherwise, we choose $a_{i+1}$ arbitrarily from one of the $\widehat{C}_k$ and set $A^{i+1}_n := A^i_n \cup \{a_{i+1}\}$. Then we proceed with the construction until $K_i = \emptyset$ or all $\widehat{C}_k$ are empty. In both cases, the constructed set is as required by the lemma.\\ 
	\\
	\\
	It remains to show: $|A_n| \leq \frac{|S_n|}{2}$, i.e.\ that the construction process consists of at most $\frac{|S_n|}{2}$ steps. We do this by defining a potential function $\Phi$ that associates with any partition $\Pp$ of $[n]$ a natural number $\leq n$ that roughly says how many further refinement steps of $\Pp$ are at most possible. Concretely:
	\[
	\Phi(\Pp) := \sum_{P \in \Pp} \max\{(|P \cap S_n| - 2),0\}.
	\]
	If $\Pp$ contains as its only part the whole set $[n]$, then $\Phi(\Pp) = |S_n|-2$. Now observe that a necessary condition for adding a new string $a_{i+1}$ to $A_n$ is the existence of a part $P$ with $|P \cap S_n| > 2$ in the current partition $\Pp = \bigsqcap A^i_n$. This is the case if and only if $\Phi(\Pp) > 0$. Therefore, all that remains to be shown is:
	\begin{equation}
		\Phi\left(\bigsqcap A^{i}_n\right) - \Phi\left(\bigsqcap A^{i+1}_n\right) \geq 2 \tag{$\star$}
	\end{equation}
	for all construction steps $i$.\\
	To this end, consider step $i+1$: We add $a_{i+1}$ to $A^i_n$. There are two (not necessarily disjoint) cases:\\
	The first case is: There are two distinct parts $P,P' \in \bigsqcap A_n^i$ such that $a_{i+1}$ is non-constant on $P \cap S_n$ and $P' \cap S_n$, and  $|P \cap S_n| > 2$ and $|P' \cap S_n| > 2$. In this case, both $P \cap S_n$ and $P' \cap S_n$ will be split when $a_{i+1}$ is added, and each of these splits reduces the potential $\Phi$ by at least one, so ($\star$) holds.\\
	In the second case, there is a part $P$ such that $a_{i+1}[P \cap S_n]$ is not imbalanced and not constant. Therefore, $P \cap S_n$ will be split into two parts $P_1,P_2$ with $|P_1 \cap S_n| \geq 2$, $|P_2 \cap S_n| \geq 2$. This means that the contribution of $P \cap S_n$ to $\Phi(\bigsqcap A_n^i)$, namely $p:= |P \cap S_n| - 2$, is reduced to $|P_1 \cap S_n| - 2 + |P_2 \cap S_n| - 2 = p-2$ (here it is important that the new parts both have size at least two). So also in this case, ($\star$) holds.
\end{proof}

\noindent Now we have almost all pieces that we need to prove a good bound on $|\Stab(B_n)|$: Given that $|B_n| \leq cn$ for some constant $c$, there are at most $(cn)^{|S_n|/2}$ possibilities to specify a $\sigma \in \Sym(B_n)$ on the elements of the set $A_n$ from the previous lemma. It remains to show that each such $\sigma$ cannot be realised by too many permutations in $\Sym_n$. If for all parts $P \in \bigsqcap A_n$, we had $|P \cap S_n| \leq 2$, then Corollary \ref{cor:LOG_centralCorollary} would already imply our desired bound. However, there may also be parts $P \in \bigsqcap A_n$ where $|P \cap S_n|$ is unbounded. The next lemma shows how the properties of these parts that are stated in Lemma \ref{lem:LOG_existenceOfSubsetA} help us to deal with them. Essentially, it says that a permutation in $\Stab(B_n)$ is already fully specified if we only know how it moves the parts $P \in \bigsqcap A_n$ with $|P \cap S_n| \leq 2$. In other words, we can count the number of permutations in $\Stab(B_n)$ by just looking at their possible behaviour on the parts where $|P \cap S_n| \leq 2$.

%TODO: noch ein Korollar machen, das lemma 9 und 16 kombiniert, und dann in den folgenden zwei Lemmas darauf verweisen!
\begin{lemma}
	\label{lem:LOG_pi}
	Let $A_n \subseteq B_n$ be the subset that exists by Lemma \ref{lem:LOG_existenceOfSubsetA}, and let $p : A_n \lra B_n$ be an injective function. Let 
	\[
	\Gamma_p := \{ \pi \in \Stab(B_n) \mid \pi(p(a)) = a \text{ for all } a \in A_n \}.
	\]
	Further, let 
	\[
	P_{>2} := \{k \in S_n \mid |P(k) \cap S_n| > 2 \text{, where } P(k) \in \bigsqcap A_n \text{ is the part that } k \text{ is in}\}.
	\]
	Then for any $\pi, \pi' \in \Gamma_p$ such that $\pi^{-1}|_{([n] \setminus P_{>2})} = \pi'^{-1}|_{([n] \setminus P_{>2})}$, it also holds $\pi^{-1}|_{P_{>2}} = \pi'^{-1}|_{P_{>2}}$.
\end{lemma}
%TODO: Problem: P(x) ist nicht unbedingt identisch zu P(k)!
\begin{proof}
	For a contradiction, we assume that there exist $\pi, \pi' \in \Gamma_p$ such that $\pi^{-1}|_{([n] \setminus P_{>2})} = \pi'^{-1}|_{([n] \setminus P_{>2})}$, but  $\pi^{-1}|_{P_{>2}} \neq \pi'^{-1}|_{P_{>2}}$. Then there is $x \in [n]$ such that $\pi(x) \in P_{>2}$, and $\pi'(x) \neq \pi(x)$ (i.e.\ $\pi(x)$ is the point where $\pi^{-1}$ and $\pi'^{-1}$ differ). Let $y := \pi(x), y' := \pi'(x)$. Let $P(y) \in \bigsqcap A_n$ be the part that $y$ is in, and let $\widehat{P}(y) := P(y) \cap S_n$. We know that $y' \in P(y)$, too, because $\pi, \pi' \in \Gamma_p$, so this follows from Lemma \ref{lem:LOG_centralLemma}. As $y \in P_{>2}$, in particular, $y \in S_n$. Hence, also $x,y' \in S_n$ because $\pi, \pi' \in \Stab(B_n)$, and by Lemma \ref{lem:sandwichLemma}, singleton parts must be mapped to singleton parts. We conclude that we even have $y' \in \widehat{P}(y)$.\\
	Now, our goal is to show that the transposition $\tau := (y \ y')$ is contained in $\Stab(B_n)$. This is a contradiction because in $\Sp(B_n)$, $y, y'$ are both in singleton parts. However, if $(y \ y') \in \Stab(B_n)$, then there is a coarser supporting partition in which $\{y, y'\}$ forms one part. This is a contradiction as $\Sp(B_n)$ is the coarsest possible support.\\
	In order to show $\tau \in \Stab(B_n)$, we only need to deal with those strings in $B_n$ which are not constant on the positions $\{y, y'\}$. More precisely, we have to show that every $b \in B_n$ with $b_y \neq b_{y'}$ has a "swapping partner" $b' \in B_n$ where $b'_y = b_{y'}$ and vice versa, and $b'_i = b_i$ for all other $i$.\\
	So take any $b \in B_n$ such that w.l.o.g.\ $b_y = 0, b_{y'} = 1$. Note that $b \notin A_n$, as every string in $A_n$ is constant on $P(y)$ (otherwise, $P(y)$ would not be a single part in $\bigsqcap A_n$). Furthermore, $|\widehat{P}(y)| > 2$, since $y \in P_{>2}$. Therefore, Lemma \ref{lem:LOG_existenceOfSubsetA} implies that the substring $b[\widehat{P}(y)]$ is imbalanced and $b$ is constant on every $P' \cap S_n$, for all $P' \in \bigsqcap A_n$ with $P' \neq P(y)$, $|P' \cap S_n| > 2$. W.l.o.g.\ let the imbalance of $b[\widehat{P}(y)]$ be such that $b_i = 1$ for every position $i \in \widehat{P}(y), i \neq y$. We claim that $b' := (\pi' \circ \pi^{-1})(b) \in B_n$ is the desired swapping partner of $b$, i.e.\ $\tau(b) = b'$ and vice versa. The strings $b$ and $(\pi' \circ \pi^{-1})(b)$ look somewhat like this:\\
	\begin{figure}[H]
		\centering
		\begin{tikzpicture}[line width = 0.5mm]
			\node at (-1,-0.4) {$b$};
			\node at (-1.2,-2.5) {$(\pi' \circ \pi^{-1})(b)$};
			%	\node at (-1.9,-2.3) {$\Sp(A_1) \sqcap \Sp(A_2)$};
			%\draw (0,0) rectangle (5,-0.75);
			%\draw[fill = blue!50] (1,0) rectangle (5,-0.75);
			
			%	\draw[fill = yellow!30] (0,-1) rectangle (3,-1.75);
			%	\draw[fill = red!70] (3,-1) rectangle (5,-1.75);
			\draw[decorate, decoration={brace}] (0,0.2) -- (4.5 ,0.2 ) node[pos=0.5, above] {$S_n$};
			
			\draw[fill=green!20]  (0,0) rectangle (4.5,-0.75);	
			\draw[fill=green!20]  (0,-2) rectangle (4.5,-2.75);	
			
			\draw  (0,0) rectangle (1,-0.75);	
			\draw  (1,0) rectangle (4,-0.75) node[pos=.5] {$0111111111111$};	
			\draw  (4,0) rectangle (5,-0.75);	
			
			\draw[red,->] (1.3,-1.2) -- (1.3,-0.5);
			\draw[red,->] (1.3,-1.6) -- (1.3,-2.2);
			\node at (1.3,-1.4) {$y$}; 
			\draw[red,->] (3.7,-1.2) -- (3.7,-0.5);			
			\draw[red,->] (3.7,-1.6) -- (3.7,-2.2);			
			\node at (3.7,-1.4) {$y'$}; 
			\node at (2.5,-1.3) {$\widehat{P}(y)$}; 
			
			\draw (1.3,-0.5) edge[blue,->, bend right=20] node [above, pos=0.6, xshift=-0.1cm] {$\pi^{-1}$} (0.5,-1.25);
			\node at (0.5,-1.4) {$x$};
			\draw (0.5,-1.5) edge[blue,->, bend right=90] node [below] {$\pi'$} (3.7,-2.6);
			
			\draw  (0,-2) rectangle (1,-2.75);	
			\draw  (1,-2) rectangle (4,-2.75) node[pos=.5] {$1111111111110$};	
			\draw  (4,-2) rectangle (5,-2.75);	
		\end{tikzpicture}
	\end{figure}
	To see that $\tau(b) = (\pi' \circ \pi^{-1})(b)$, consider firstly $\pi^{-1}(b) \in B_n$. Obviously, $(\pi^{-1}(b))_x = 0$. The string $(\pi' \circ \pi^{-1})(b)$ is also in $B_n$ and we have $(\pi' \circ \pi^{-1})(b)_{y'} = 0$. Moreover, the substring $\pi^{-1}(b)[\pi^{-1}(\widehat{P}(y))]$ is imbalanced just like $b[\widehat{P}(y)]$, so $(\pi^{-1}(b))_j = 1$ for all $j \in \pi^{-1}(\widehat{P}(y)) \setminus \{x\}$. As a consequence of Corollary \ref{cor:LOG_centralCorollary}, we have $\pi^{-1}(\widehat{P}(y)) = \pi'^{-1}(\widehat{P}(y))$. Therefore, the substring $(\pi' \circ \pi^{-1})(b)[\widehat{P}(y)]$ is also imbalanced and has a $1$ at each position except $y'$.\\
	This shows that $(\tau(b))[\widehat{P}(y)] = b'[\widehat{P}(y)]$. It remains to show that $b_i = b'_i$ for all $i \in [n] \setminus \widehat{P}(y)$.\\
	We have $(\pi' \circ \pi^{-1})(i) = i$ for $i \in [n] \setminus P_{>2}$, because $\pi^{-1}|_{([n] \setminus P_{>2})} = \pi'^{-1}|_{([n] \setminus P_{>2})}$, so $b_i = b'_i$ for $i \in [n] \setminus P_{>2}$.\\
	For $i \in P_{>2} \setminus \widehat{P}(y)$, let $\widehat{P}(i)$ be the part of $\bigsqcap A_n$ that $i$ is in, intersected with $S_n$. As already said, we know from Lemma \ref{lem:LOG_existenceOfSubsetA} that $b$ is constant on $\widehat{P}(i)$. Analogously to what we argued already for $\widehat{P}(y)$, we get that $(\pi' \circ \pi^{-1})(\widehat{P}(i)) = \widehat{P}(i)$, so also for $i \in P_{>2} \setminus \widehat{P}(y)$, we have $b_i = b'_i$.\\
	In total, this shows that indeed, $b' = \tau(b)$, and since $b' \in B_n$, we have $\tau \in \Stab(B_n)$. This is a contradiction and finishes the proof of the lemma.	
\end{proof}

With this, we can finally compute our upper bound on $|\Stab(B_n)|$.
\begin{lemma}
	\label{lem:LOG_stabilizerBound}	
	Let $c$ be a constant such that $|B_n| \leq c \cdot n$ (for large enough $n$). Then, for large enough $n$, it holds:
	\[
	|\Stab(B_n)| \leq (2cn)^{|S_n|/2} \cdot (n-|S_n|)!
	\]	
\end{lemma}
\begin{proof}
	Let $A_n \subseteq B_n$ be the subset of $B_n$ whose existence is stated in Lemma \ref{lem:LOG_existenceOfSubsetA}.
	Fix any injective function $p : A_n \lra B_n$. As in the previous lemma, let
	\[
	\Gamma_p := \{ \pi \in \Stab(B_n) \mid \pi(p(a)) = a \text{ for all } a \in A_n \}.
	\]
	We bound $|\Gamma_p|$ by counting the number of possible $\pi \in \Gamma_p$. We know by Lemma \ref{lem:LOG_pi} that we only have to count the number of possibilities to choose the preimages of the elements in $[n] \setminus P_{>2}$, where again,
	\[
	P_{>2} := \{k \in S_n \mid |P(k) \cap S_n| > 2 \text{, where } P(k) \in \bigsqcap A_n \text{ is the part that } k \text{ is in}\}.
	\]
	For every part $P \in \bigsqcap A_n$ with $|P \cap S_n| \leq 2$, we know by Corollary \ref{cor:LOG_centralCorollary} that $\pi^{-1}(P \cap S_n) \subseteq S_n$ is the same fixed set of size $\leq 2$ for all $\pi \in \Gamma_p$, so we only have two options how $\pi^{-1}$ can behave on $P \cap S_n$. The number of such parts $P$ is at most $|S_n|/2$.\\
	\\
	For $i \in [n] \setminus S_n$, we can only say that $\pi^{-1}(i) \notin S_n$ (by Lemma \ref{lem:sandwichLemma}). Hence, every $\pi \in \Gamma_p$ can in principle permute the set $[n] \setminus S_n$ arbitrarily. In total, we conclude:
	\[
	|\Gamma_p| \leq 2^{(|S_n|/2)} \cdot (n-|S_n|)!
	\]
	This is for a fixed function $p$. The number of possible choices for $p$ is bounded by $(cn)^{|S_n|/2}$, since $|A_n| \leq |S_n|/2$ (Lemma \ref{lem:LOG_existenceOfSubsetA}) and we are assuming $|B_n| \leq cn$.\\
	Every $\pi \in \Stab(B_n)$ must occur in at least one of the sets $\Gamma_p$ for some choice of $p$, so indeed, $(2cn)^{|S_n|/2} \cdot (n-|S_n|)!$ is an upper bound for $|\Stab(B_n)|$. 
\end{proof}
As in the case where $|S_n|$ grows sublinearly in $n$, we conclude this part of the proof by computing the orbit size of $B_n$ based on the stabiliser bound and comparing this to any polynomial in $2^n$.

\begin{lemma}
	\label{lem:orbitCase12}
	Let $c$ be a constant such that $|B_n| \leq c \cdot n$, and $\delta_s > 0$ be a constant such that  $|S_n| \geq \delta_s \cdot n$ (for large enough $n$).
	Then for any $k \in \bbN$, the limit 
	\[
	\lim_{n\to\infty} \frac{|\Orbit(B_n)|}{2^{kn}} = \lim_{n\to\infty} \frac{n!}{|\Stab(B_n)| \cdot 2^{kn}} 
	\]
	does not exist. That is to say, the orbit of $B_n$ w.r.t.\ the action of $\Sym_n$ on $\{0,1\}^n$ grows super-polynomially in $2^n$.
\end{lemma}
\begin{proof}
	Plugging in $|S_n| \geq \delta_s \cdot n$ into the stabiliser-bound from Lemma \ref{lem:LOG_stabilizerBound} yields:
	\[
	|\Stab(B_n)| \leq (2cn)^{(\delta_s n)/2} \cdot ((1-\delta_s)n)!
	\]
	This is true because larger values for $|S_n|$ only make the expression $2^{|S_n|/2} \cdot (n-|S_n|)!$ smaller.\\
	Replacing factorials with the Stirling Formula, the fraction that we are taking the limit of becomes at least:
	\begin{align*}
		&\frac{1}{2\sqrt{1-\delta_s}} \cdot \left(\frac{1}{(1-\delta_s)}\right)^{(1-\delta_s)n} \cdot \left(\frac{n}{e}\right)^{\delta_s n} \cdot \frac{1}{(2cn)^{(\delta_s n)/2}} \cdot \frac{1}{2^{kn}}\\
		= &\frac{1}{2\sqrt{1-\delta_s}} \cdot \left(\frac{1}{2^k(1-\delta_s)}\right)^{(1-\delta_s)n} \cdot \left(\frac{n}{2c \cdot 4^k \cdot e^2}\right)^{(\delta_s n)/2}\\
		= &\frac{1}{2\sqrt{1-\delta_s}} \cdot \left(\frac{n}{(2^k(1-\delta_s))^{2(1-\delta_s)/\delta_s } \cdot 2c \cdot 4^k \cdot e^2}\right)^{(\delta_s n)/2}
	\end{align*}
	The denominator is constant, so the limit does not exist.	
\end{proof}
\noindent
This lemma together with Lemma \ref{lem:LOG_orbitCase11} proves Lemma \ref{lem:LOG_caseLinearSupport}.\\
Theorem \ref{thm:mainTechnical} now follows directly from Lemmas \ref{lem:resultCaseSublinear} and \ref{lem:LOG_caseLinearSupport}, since these two lemmas cover all cases.

\section{Concluding remarks and future research}
\label{sec:conclusion}
A question that remains open is what exactly is the threshold of ``fineness'' of a preorder where the orbit size changes from super-polynomial to polynomial. In other words: What is the largest colour class size for which our super-polynomial orbit theorem for hypercubes still holds?\\
One can check that all parts of our proof can be modified such that it also goes through if we allow colour classes of size $o(n^2)$. If the size is in $\Theta(n^2)$, though, the bound in \autoref{lem:LOG_stabilizerBound} becomes too large for \autoref{lem:orbitCase12} to hold.\\
On the other hand, the finest preorder with a polynomial orbit that we know so far is one where the colour class sizes are in $\Oo(2^n/\sqrt{n})$: It corresponds to the partition of $\{0,1\}^n$ according to Hamming-weight. This is precisely the orbit-partition of the vertex-set (w.r.t.\ the action of $\Sym_n$ on the positions). Its largest colour class has size $\binom{n}{n/2} \in \Theta(2^n/\sqrt{n})$.\\
Determining the finest preorder that is in principle CPT-definable in hypercubes would potentially allow to better judge whether a preorder-based CPT-algorithm like the one in \cite{pakusa2018definability} can at all be a candidate for a solution of the unordered CFI problem.\\
Moreover, it would be helpful to identify further h.f.\ objects that are undefinable in hypercubes for symmetry reasons.

%TODO: eventuell die DOIs raussuchen für alle Quellen
\bibliography{references.bib}

\begin{thebibliography}{10}

\bibitem{anderson2017symmetric}
Matthew Anderson and Anuj Dawar.
\newblock On symmetric circuits and fixed-point logics.
\newblock {\em Theory of Computing Systems}, 60(3):521--551, 2017.
\newblock \href {https://doi.org/10.1007/s00224-016-9692-2}
  {\path{doi:10.1007/s00224-016-9692-2}}.

\bibitem{blass1999}
Andreas Blass, Yuri Gurevich, and Saharon Shelah.
\newblock Choiceless polynomial time.
\newblock {\em Annals of Pure and Applied Logic}, 100(1-3):141--187, 1999.
\newblock \href {https://doi.org/10.1016/S0168-0072(99)00005-6}
  {\path{doi:10.1016/S0168-0072(99)00005-6}}.

\bibitem{caifurimm92}
Jin-yi Cai, Martin F{\"u}rer, and Neil Immerman.
\newblock An optimal lower bound on the number of variables for graph
  identification.
\newblock {\em Combinatorica}, 12:389--410, 1992.
\newblock \href {https://doi.org/10.1007/BF01305232}
  {\path{doi:10.1007/BF01305232}}.

\bibitem{chandra1980structure}
Ashok~K Chandra and David Harel.
\newblock Structure and complexity of relational queries.
\newblock In {\em 21st Annual Symposium on Foundations of Computer Science
  (sfcs 1980)}, pages 333--347. IEEE, 1980.
\newblock \href {https://doi.org/10.1109/SFCS.1980.41}
  {\path{doi:10.1109/SFCS.1980.41}}.

\bibitem{dawar2015nature}
Anuj Dawar.
\newblock The nature and power of fixed-point logic with counting.
\newblock {\em ACM SIGLOG News}, 2(1):8--21, 2015.

\bibitem{dawar2009logics}
Anuj Dawar, Martin Grohe, Bjarki Holm, and Bastian Laubner.
\newblock Logics with rank operators.
\newblock In {\em 2009 24th Annual IEEE Symposium on Logic In Computer
  Science}, pages 113--122. IEEE, 2009.
\newblock \href {https://doi.org/10.1109/LICS.2009.24}
  {\path{doi:10.1109/LICS.2009.24}}.

\bibitem{dawar2008choiceless}
Anuj Dawar, David Richerby, and Benjamin Rossman.
\newblock Choiceless polynomial time, counting and the
  {C}ai--{F}{\"u}rer--{I}mmerman graphs.
\newblock {\em Annals of Pure and Applied Logic}, 152(1-3):31--50, 2008.
\newblock \href {https://doi.org/10.1016/j.apal.2007.11.011}
  {\path{doi:10.1016/j.apal.2007.11.011}}.

\bibitem{gradel2015polynomial}
Erich Gr{\"a}del and Martin Grohe.
\newblock Is {P}olynomial {T}ime {C}hoiceless?
\newblock In {\em Fields of Logic and Computation II}, pages 193--209.
  Springer, 2015.
\newblock \href {https://doi.org/10.1007/978-3-319-23534-9\_11}
  {\path{doi:10.1007/978-3-319-23534-9\_11}}.

\bibitem{grapakschalkai15}
Erich Gr{\"{a}}del, Wied Pakusa, Svenja Schalth{\"{o}}fer, and {\L}ukasz
  Kaiser.
\newblock Characterising {C}hoiceless {P}olynomial {T}ime with {F}irst-order
  {I}nterpretations.
\newblock In {\em Proceedings of the 30th Annual {ACM/IEEE} Symposium on Logic
  in Computer Science}, pages 677--688, 2015.
\newblock \href {https://doi.org/10.1109/LICS.2015.68}
  {\path{doi:10.1109/LICS.2015.68}}.

\bibitem{grohe2008quest}
Martin Grohe.
\newblock The quest for a logic capturing {PTIME}.
\newblock In {\em 2008 23rd Annual IEEE Symposium on Logic in Computer
  Science}, pages 267--271. IEEE, 2008.
\newblock \href {https://doi.org/10.1109/LICS.2008.11}
  {\path{doi:10.1109/LICS.2008.11}}.

\bibitem{gurevich1985logic}
Yuri Gurevich.
\newblock Logic and the {C}hallenge of {C}omputer {S}cience.
\newblock In {\em Current Trends in Theoretical Computer Science}. Computer
  Science Press, 1988.

\bibitem{harary2000automorphism}
Frank Harary.
\newblock The automorphism group of a hypercube.
\newblock {\em J. Univers. Comput. Sci.}, 6(1):136--138, 2000.

\bibitem{immerman1982relational}
Neil Immerman.
\newblock Relational queries computable in polynomial time.
\newblock In {\em Proceedings of the fourteenth annual ACM symposium on Theory
  of computing}, pages 147--152, 1982.
\newblock \href {https://doi.org/10.1145/800070.802187}
  {\path{doi:10.1145/800070.802187}}.

\bibitem{pakusa2015linear}
Wied Pakusa.
\newblock {\em Linear Equation Systems and the Search for a Logical
  Characterisation of Polynomial Time}.
\newblock PhD thesis, RWTH Aachen, 2015.

\bibitem{pakusa2018definability}
Wied Pakusa, Svenja Schalth{\"o}fer, and Erkal Selman.
\newblock Definability of {C}ai-{F}{\"u}rer-{I}mmerman problems in {C}hoiceless
  {P}olynomial {T}ime.
\newblock {\em ACM Transactions on Computational Logic (TOCL)}, 19(2):1--27,
  2018.
\newblock \href {https://doi.org/10.1145/3154456} {\path{doi:10.1145/3154456}}.

\bibitem{rossman2010choiceless}
Benjamin Rossman.
\newblock Choiceless computation and symmetry.
\newblock In {\em Fields of Logic and Computation, Essays Dedicated to Yuri
  Gurevich on the Occasion of His 70th Birthday}, volume 6300 of {\em Lecture
  Notes in Computer Science}, pages 565--580. Springer, 2010.
\newblock \href {https://doi.org/10.1007/978-3-642-15025-8\_28}
  {\path{doi:10.1007/978-3-642-15025-8\_28}}.

\bibitem{svenja}
Svenja Schalthöfer.
\newblock {\em Choiceless Computation and Logic}.
\newblock PhD thesis, RWTH Aachen, 2020.

\bibitem{vardi1982complexity}
Moshe~Y Vardi.
\newblock The complexity of relational query languages.
\newblock In {\em Proceedings of the fourteenth annual ACM symposium on Theory
  of computing}, pages 137--146, 1982.
\newblock \href {https://doi.org/10.1145/800070.802186}
  {\path{doi:10.1145/800070.802186}}.

\end{thebibliography}

\end{document}